\documentclass[aps,pra,notitlepage]{revtex4-1}
\input epsf
\usepackage{graphicx}
\usepackage{bm}
\usepackage{amsmath,amssymb,amsfonts,amsthm,amscd}
\usepackage{tikz}
\usepackage{hyperref}

\newcommand{\beq}{\begin{equation}}
\newcommand{\eeq}{\end{equation}}
\newcommand{\beqa}{\begin{eqnarray}}
\newcommand{\eeqa}{\end{eqnarray}}
\newcommand{\beqan}{\begin{eqnarray*}}
\newcommand{\eenan}{\end{eqnarray*}}

\def\Bbb{\mathbb}

\newcommand{\Dslash}{{\slash{\kern -0.5em}\partial}}
\newcommand{\Aslash}{{\slash{\kern -0.5em}A}}

\def\sqr#1#2{{\vcenter{\hrule height.#2pt
     \hbox{\vrule width.#2pt height#1pt \kern#1pt
        \vrule width.#2pt}
     \hrule height.#2pt}}}

\def\thinspace{\kern .16667em}

\def\xp{x_{{\kern -.2em}_\perp}}
\def\subp{_{{\kern -.2em}_\perp}}

\def\defeq{:{\kern -0.5em}=}
\def\inlinedefeq{:{\kern -0.8em}=}
\def\Tr{{\rm Tr}\,}

\def\JN{{\mathcal J}_N}

\def\X{{\mathcal X}}

\def\XP{{\mathcal X}^+}
\def\XPP{{\mathcal X}^{++}}
\def\XNP{{\mathcal X}_N^+}
\def\XNPP{{\mathcal X}_N^{++}}
\def\Part{{\mathfrak P}}
\def\Fmax{F_\mathrm{max}}
\def\Fhat{\hat{F}}
\def\V{\mathcal V}

\def\dotro{\varrho}
\def\scale{\sf{scale}}

\def\dom{\mathrm{dom}\,}

\def\epi{\mathrm{epi}\,}
\def\graph{\mathrm{graph}\,}

\def\rest{{\kern -0.3em}\upharpoonright}

\newtheorem{lem}{Lemma}[section]
\newtheorem{thm}{Theorem}[section]
\newtheorem{cor}{Corollary}[section]
\newtheorem*{cor*}{Corollary}
\newtheorem*{thm*}{Theorem}
\newtheorem{prop}{Proposition}[section]

\begin{document}

\title{Taming Density Functional Theory by Coarse-Graining}
\author{Paul~E.~Lammert}
\affiliation{
Department of Physics, 104B Davey Lab \\
Pennsylvania State University \\
University Park, PA 16802-6300 \\
pel1@psu.edu}

\begin{abstract}
The standard (``fine-grained'') interpretation of quantum density functional theory,
in which densities are specified with infinitely-fine spatial resolution,
is mathematically unruly.  Here,
a coarse-grained version of DFT, featuring limited spatial resolution, and its
relation to the fine-grained theory in the $L^1\cap L^3$ formulation of
Lieb, is studied, with the object of showing it to be not only mathematically
well-behaved, but consonant with the spirit of DFT, practically (computationally)
adequate and sufficiently close to the standard interpretation as to accurately
reflect its non-pathological properties.
The coarse-grained interpretation is shown to be a good model of formal DFT
in the sense that:
all densities are (ensemble)-V-representable; the intrinsic energy functional $F$
is a continuous function of the density and the representing external potential is 
the (directional) functional derivative of the intrinsic energy.
Also, the representing potential $v[\rho]$ is quasi-continuous, in that
$v[\rho]\rho$ is continuous as a function of $\rho$.
The limit of coarse-graining scale going to zero is studied to see if convergence
to the non-pathological aspects of the fine-grained theory is adequate to justify
regarding coarse-graining as a good approximation.
Suitable limiting behaviors or intrinsic energy, densities and representing potentials
are found.  Intrinsic energy converges monotonically, coarse-grained densities
converge uniformly strongly to their low-intrinsic-energy fine-grainings, and
$L^{3/2}+L^\infty$ representability of a density is equivalent to the existence
of a convergent sequence of coarse-grained potential/ground-state density pairs.
\end{abstract}
\maketitle


\section{Introduction}

Underlying electronic density functional theory 
(DFT)\cite{Parr-Yang,Dreizler-Gross,Eschrig,Martin,Capelle06} in both the
dominant Kohn-Sham\cite{KS65} and Orbital-Free\cite{Ho08,Chai07,Garcia-Cervera07}
variants, is a density functional, $F[\rho]$, representing the minimum kinetic plus 
(Coulomb) interaction energy of $N$ electrons compatible with the density $\rho$.
Clearly, the properties of $F$ are of great importance, 
and the formal development of DFT seems, at least implicitly, to involve 
assumptions that
({\it a}) the intrinsic energy functional is differentiable,
({\it b}) each density can be selected as a ground state density of an
external potential, which is the functional derivative of $F$ at $\rho$,
({\it c}) $F$ is continuous.
None of these is true.  Better to say: none of them holds in the 
standard interpretation.  For, there are things to be interpreted.
Continuity requires a topology, and derivatives come in various types.
Might the interpretational task extend even to the terms
`density' and `potential'?
The standard interpretation is {\it fine-grained} in that a density is 
assigned to every point (``almost-every'' point, technically).  
If we construct a regular partition of space into cells and regard 
number-in-cell divided by cell-volume as `{density}', then we 
arrive at a {\it coarse-grained}, or resolution-limited, interpretation.  
The thesis of this paper is that this coarse-grained model is consonant with
the spirit of DFT, mathematically benign, and a good approximation 
to the fine-grained theory in that it faithfully reflects its 
non-pathological aspects.
The coarse-grained interpretation renders every density the ground state
of some potential, essentially uniquely (Hohenberg-Kohn theorem), it satisfies 
assumptions ({a-c}) above, and the representing potential $v[\rho]$ is 
`quasi-continuous' in that $v[\rho]\rho$ is $L^1$ continuous.  The good 
approximation properties
are phrased in terms of limits as the coarse-graining scale is taken to
zero.  Intrinsic energy converges monotonically to its fine-grained value.
A fine-grained density is V-representable by a potential in $L^{3/2}+L^\infty$
if and only if there is a sequence of coarse-grained densities which
converge in an appropriate way, along with their representing potentials,
to the fine-grained data.  Coarse-grained densities approximate in a uniform
way all the fine-grained densities of low intrinsic energy which project to them.
 
To expose a little better the basic idea and its consonance with DFT, consider the 
constrained-search formulation\cite{Levy79,Levy82,Valone80,Lieb83,Argaman02}.
The focus of our interest is the Valone-Lieb intrinsic energy 
functional\cite{Valone80,Lieb83} defined as
\beq
{F}[\rho] \defeq \inf \left\{\Tr \Gamma(T+V_{ee}) :\, \Gamma \mapsto \rho  
\right\}.
\label{F1}
\eeq
Minimization is over antisymmetric mixed states of $N$ identical fermions
$\Gamma = \sum_i \lambda_i |\psi_i\rangle\langle \psi_i|$
with $\sum_i \lambda_i = 1$ having single-particle density $\rho$
(this relation is denoted $\Gamma \mapsto \rho$).
$T$ is the kinetic energy and $V_{ee}$ the interaction energy among the fermions,
so that $\Tr \Gamma(T+V_{ee})$ is the expectation value of kinetic-plus-interaction 
energy in the state $\Gamma$. 
We have in mind electrons in particular, and the interaction is assumed 
no more singular than Coulomb.
States of the system are partitioned into equivalence classes according
to their density.  Ultimately, only the lowest intrinsic energy state in each 
equivalence class matters (assuming that the minimizer exists).
For external single-particle potentials $v$ in some 
set ${\V}$, we define the ground-state energy as a Legendre-Fenchel 
transform of $F$:
\beq
E[v] \defeq \inf \left\{ F[\rho] + \int v({\bm x}) \rho({\bm x})\, d{\bm x} \, : \, 
\rho \in {\mathcal D} \right\}.
\label{LF-transform}
\eeq
If $T + V_{ee} + v$ really has a ground state with a density in ${\mathcal D}$, 
this will pick out its energy.
The point is that the potential couples to the density only, so that the
minimization problem ``find the ground state'' can be split: 
first find the minimizing $\rho$ in Eq. (\ref{LF-transform}), then
go back and identify the state from Eq. (\ref{F1}).  

Suppose now a partition $\Part$ of space into disjoint cells, and 
take ${\V}$ to consist of potentials uniform over each cell
($\Part$-measurable).
In that case, the potentials are sensitive to the average density
over a cell, but not to the variation within; that variation still
exists, but we relinquish control over it.  Densities fall into 
equivalence classes, {\it coarse-grained densities}, 
identified by average density over each cell.
Ultimately, only densities (possibly nonunique) of lowest intrinsic 
energy in each equivalence class matter.  
With ${\V}$ thus adapted to $\Part$, we make
a reinterpretation, taking ${\mathcal D}$ to consist of 
{\it coarse-grained\/} densities,
and Eqs. (\ref{F1}-\ref{LF-transform}) apply {\em as they stand}.
This is the basic idea of coarse-graining.  We are {\em not\/} restricting
the densities or the states or changing the quantum mechanics by
working on a lattice, for example, but are working with a limited
spatial resolution, and thus cannot distinguish all densities.  
It is in the nature of the problem that the fine-grained member of 
each coarse-grained density which is selected is the one with
lowest intrinsic energy.  The projective character of the coarse-graining
on the density side of the ledger is naturally paired with a restriction
on the potentials.  The partition $\Part$ has logical priority. 
It determines a collection of number observables which in turn
determine the denotations of `density' and `potential'.

Some densities are ground-state densities in some external potential
(V-representable).  If $v$ is one such potential, adding an arbitrary 
constant to $v$ preserves this relationship.  
The Hohenberg-Kohn theorem (which holds also in the coarse-grained theory,
see Section \ref{CG}) says that this is the only freedom.
We establish a convention for fixing the constant: 
$v[\rho]$ denotes the potential having $\rho$ as a ground-state density
with the constant adjusted so that
\beq
F[\rho] + \langle v, \rho\rangle \defeq
F[\rho] + \int v({\bm x})\rho({\bm x})\, d{\bm x} =0.
\label{gauge-convention}
\eeq
Equivalently, $E[v[\rho]] = 0$.
This particular convention is somewhat arbitrary, but some such in necessary for
all the statements in the introductory paragraph to make sense.

The body of the paper unfolds as follows.
Section \ref{fine-grained} discusses relevant aspects of the fine-grained, 
theory.  In particular, the ``bad-behavior'' alluded to earlier,
which is not so well-known, is exposed.  This provides both motivation 
and context for the coarse-grained model.  A conclusion is that the
difficulties are connected to short distance scales, which suggests
that some sort of short-distance regularization is called for.
As Chayes, Chayes and Ruskai\cite{CCR85} put it, ``long-distance difficulties 
do not generally occur in HK theory.''  But short-distance ones
certainly do.  Lattice models\cite{Kohn83,CCR85,Ullrich05} are one means of
short-distance regularization.  But such discretizations, if taken seriously, 
involve fundamental changes to quantum mechanics.  Their relation to the 
continuum theory, and the continuum limit, is tricky and ambiguous.  
The projective character of the regularization achieved by coarse-graining is 
much gentler.  One might say that the modification is at an epistemic rather
than ontic level.  It is for that reason that we use the somewhat awkward
term `fine-grained' rather than `continuum'.  The continuum nature of
space is recognized by the coarse-grained interpretation.
Some technical results are also developed for use in Section \ref{continuum}.

As for lattice formulations, approximations which use a limited basis for states or 
density matrices (studied in great detail in the series of 
papers \cite{Harriman-GDM1,Harriman-GDM2,Harriman-GDM3,Harriman-GDM4}) must
be differentiated from the current approach.
A coarse-grained density represents an equivalence class of infinitely many densities.  
Insofar as low energies are of interest, most of them are not very relevant, but the 
selection is an energetic one, not an {\it a priori} choice.
Coupled with the {\it projective} treatment of densities in the coarse-grained
model is an {\it injective} approximation of densities; it is here that a
functional palette is limited.  This might seem a somewhat uncomfortable aspect.  
We would like to handle Coulomb potentials of point charges, for instance.  
But an atomic nucleus is not a point charge and if we treat it as such, we are 
half a step from coarse-graining at the femtometer scale, anyway.  

The coarse-grained model is set up in Section \ref{CG} and results from a 
previous paper\cite{Lammert06} are reviewed.  The new results for the single-scale 
coarse-grained model are discussed and proven in Section \ref{single-scale}, showing that
it is free of the bad behavior which afflicts the fine-grained theory.  
Theorem \ref{Continuity-Thm} shows essential continuity of the intrinsic energy,
Theorem \ref{Diff-Thm} shows that physical directional derivatives of $F[\rho]$ coincide
with representing potentials $v[\rho]$ and Theorem \ref{V-Cont-Thm} shows that 
$v[\rho]\rho$ is continuous.  Kohn-Sham theory is taken up in Section \ref{KS},
where the exchange-correlation potential is shown to exist at least as the
directional derivative of the exchange-correlation energy.

Sections \ref{CG} through \ref{KS} show that the single-scale coarse-grained model 
is a well-behaved model of formal DFT.  Its acceptability, however, depends to
some extent on its ability to reasonably approximate the fine-grained theory
as the coarse-graining scale goes to zero.  This is the subject of Section \ref{continuum},
where a multi-scale coarse-grained model is set up.  Taking the coarse-graining
scale to zero is naturally thought of as a process carried out in time,
and the multi-scale model is ultimately little other than this process 
viewed {\it sub specie aeternitatis}.  It is a convenient tool however,
and serves to remove the impression that a fundamental length scale
is inherent in the coarse-grained interpretation. 
A basic result is that all the intrinsic energy of a fine-grained density
is recovered monotonically in the limit.  
Representability of a density by a $L^{3/2}+ L^\infty$ potential 
is reliably and faithfully signalled by the coarse-grained model.
Proximity of coarse-grained densities to the low-intrinsic-energy 
fine-grained densities in the equivalence classes they represent 
is also examined.  Among other things, these results are taken 
as validating the claim that the coarse-grained model approximates
what it is supposed to.

A nodding familiarity with basic Banach space functional analysis is presumed in 
the body of the paper.  Appendix \ref{fa-appendix} contains a brief review of 
concepts, definitions and notations which may be helpful to readers needing 
a quick reminder, or more.

A couple of delimiting remarks are in order before we begin.
In this paper, mixed states are always allowed.  So, all statements about
V-representability refer specifically to {\em ensemble}-V-representability.
The results are applicable to abelian spin-density-functional theory, where
up-spin and down-spin are roughly separate species.
However, genuine non-abelian spin-density-functional 
theory,\cite{Gunnarsson-Lundqvist76,Ullrich05,Gidopoulos07}
treating all components of the spin, is unfortunately well beyond the scope of this work. 

\section{Fine-Grained DFT and its Discontents}
\label{fine-grained}

This section is an essay on fine-grained DFT, focussing on
the difficulties mentioned in the Introduction.
Lieb's landmark 1983 paper\cite{Lieb83} formulates the theory in the
precise form in which we will consider it.
The recent article\cite{vanLeeuwen03} of van Leeuwen and the book\cite{Eschrig} 
by Eschrig are also recommended, though some assertions found in them are here rejected.
Further background on convex and nonsmooth analysis can be sought
in the books \cite{ET,Aubin-Ekeland,Borwein-Zhu,Phelps88}.
Chapter 5 of Aubin and Ekeland\cite{Aubin-Ekeland} contains an interesting and 
idiosyncratic discussion of the Ekeland variational principle. 
Appendix \ref{fa-appendix} contains a refresher on functional analysis 
and Lebesgue spaces which is relevant.

It is demonstrated in \S 1 of Lieb's paper\cite{Lieb83} that the 
effective domain of $F$ is
\beq
{\mathcal J}_N \defeq \left\{ \rho: \rho({\bm x}) \ge 0, \, 
\rho^{1/2} \in H^1({\Bbb R})^3, \,
\int \rho({\bm x})\, d{\bm x} = N \right\}.
\eeq
Here, $H^1({\Bbb R}^3)$ is the Sobolev space of functions $f$
such that $\int |f|^2 + |\nabla f({\bm x})|^2 \, d{\bm x} < \infty$.
The gradient can be interpreted in terms of Fourier transform;
classical differentiability is not a requirement.
The remarkable fact which makes $\JN$ the right domain is that there
is an $N$-particle wavefunction $\psi$ with finite kinetic energy
and single-particle density $\rho$ {\em if and only if\/} 
$\rho \in {\mathcal J}_N$\cite{Harriman81,Lieb83,Zumbach-Maschke,Zumbach85}. 
``Finite kinetic energy'' is understood in quadratic form sense, i.e.,
$\sum_{i=1}^N \int |\nabla_i \psi|^2 < \infty$, so 
$\nabla^2 \psi$ does not need to be square-integrable, just $\nabla \psi$.
This condition is sufficient to guarantee that also the Coulomb interaction 
energy is finite.  
Some perturbation, for example by an arbitrary bounded interaction,
would not affect conclusions.
We assume without loss that the interaction is bounded below by zero.
Since ${\mathcal J}_N$ is convex\cite{Lieb83}, forming mixtures will not 
produce densities outside ${\mathcal J}_N$, so it is the right domain
whether considering mixed states or pure states.
(A set $X$ is convex if, whenever $x$ and $y$ are in $X$, so
is the line segment $\lambda x + (1-\lambda)y$, $0\le \lambda \le 1$.)

$F[\rho]$ is now defined as in Eq. (\ref{F1}), with ${\mathcal D}=\JN$.
$F$ has two important properties which allow the analysis to move forward:
\renewcommand{\theenumi}{\alph{enumi}}
\begin{enumerate}
\item 
$F$ is convex: 
$F[s \rho + (1-s)\rho^\prime] \le sF[\rho] + (1-s)F[\rho^\prime]$ for $0\le s \le 1$. 
\item 
$F$ is lower semicontinuous with respect to the $L^1$ topology.
\end{enumerate}
Assertions about ways in which nearby densities 
are similar requires a well-defined notion of `nearby,' i.e., a topology
(see Appendix \ref{fa-appendix}).
Different topologies can be appropriate and useful in different ways.
The statement about lower semicontinuity means:
Given $\rho \in {\mathcal J}_N$ and $\epsilon > 0$, there is
$d > 0$ such that $\|\rho^\prime - \rho\|_{1} < d$ implies that
$F[\rho^\prime] > F[\rho] - \epsilon$.  
The name should now be clear.  {\em If\/} the conclusion had instead been
$F[\rho^\prime] < F[\rho] + \epsilon$, that would be {\em upper} semicontinuity of $F$.
Continuity is the conjuction of lower and upper semicontinuity.

Why is $F$ not upper semicontinuous?  That is connected with the 
important inequalies\cite{Lieb83}
\beq 
c |\rho^{1/2}|_{H^1}^2 \le F[\rho] \le c^\prime + c^{\prime\prime} |\rho^{1/2}|_{H^1}^2, 
\label{Lieb-sandwich}
\eeq 
where the squared Sobolev seminorm is 
\[
|\rho^{1/2}|_{H^1}^2 = \int |\nabla \rho^{1/2}|^2 \, d{\bm x}
= \int \frac {|\nabla \rho|^2}{4\rho} \, d{\bm x},
\]
which will also be recognized as the von Weiszacker term\cite{Dreizler-Gross}.
The constants in the upper bound can depend on the interaction, 
but the lower bound depends only on the kinetic energy.
It is a precise form of the physical intuition that strong oscillations in the 
density cost kinetic energy.  Consider, for instance the density
\hbox{$\rho^\prime({\bm x}) \sim [1+\epsilon \eta({\bm x})\sin x/\ell]^2\rho({\bm x})$},
where $0 \le \eta \le 1$ is a smooth function with bounded support.
$\|\rho^\prime - \rho\|_1 = {\mathcal O}(\epsilon)$, but
$|{\rho^\prime}^{1/2}|_{H^1}^2 = {\mathcal O}(\epsilon^2/\ell^2)$
as $\ell \to 0$.  This is the first serious short-distance difficulty. 
Such problems are a major motivation for the development of the
coarse-grained approach.  
Although this discussion was framed with respect to $L^1$ for concreteness, 
it is easy to see that it also applies to $L^1 \cap L^p$ for $1 < p \le 3$.

Now we embed ${\mathcal J}_N$ in a Banach space $X$ with a norm at
least as strong as $L^1$ (meaning the norm dominates some fixed multiple 
of the $L^1$-norm).
Lieb noticed that a Sobolev inequality 
\beq
\|\rho\|_3 \le c (|\rho^{1/2}|_{H^1}^2 + \|\rho\|_1)
\label{L3-Sobolev}
\eeq
combines with (\ref{Lieb-sandwich}) to show that
${\mathcal J}_N$ is actually contained in $L^3$ as well as $L^1$.
Consequently, he chose $X = L^1\cap L^3$ with norm 
$\|\cdot \|_X = \|\cdot\|_1 + \| \cdot \|_3$.
But there is freedom here.  Any $L^1 \cap L^p$ with $1 < p \le 3$ 
would work, as would $L^1$.  However, a norm at least as strong as
$L^1$ is important to maintain lower semicontinuity of $F$.
The move to a Banach space puts us in position to use some relevant ideas and 
results of convex analaysis, and is suggested by simple observation that
an external potential $v$ acts as a linear functional
$\langle v,\rho \rangle \defeq \int v({\bm x}) \rho({\bm x}) \, d{\bm x}$.
$F$ is extended by taking $F \equiv +\infty$ off ${\mathcal J}_N$.
This preserves convexity and lower semicontinuity\cite{Lieb83}.
The subset of $X$ on which $F < +\infty$ is called the effective domain of $F$ 
and denoted $\dom F$.  Of course, $\dom F$ is simply $\JN$, but we will use the former
notation sometimes because the effective domain in the coarse-grained 
model is different, but plays the same role.  
Note that $\JN$ is contained in the closed affine hyperplane 
\[
X_N \defeq \left\{ \rho \in X: \int \rho({\bm x}) \, d{\bm x} = N \right\}.
\]
Since elements of $X$, and even of $X_N$, can be negative, we refer to them 
generally as `{quasi-densities},' reserving `{density}' for the non-negative elements.

The ground-state energy $E[v]$ is defined according to Eq. (\ref{LF-transform}).
But, for which potentials?  The expression on the right-hand side 
of Eq. (\ref{LF-transform}) makes sense for many potentials, and we will return
to that point.
But insofar as we want $F[\rho]$ and $E[v]$ to form a Legendre-Fenchel 
transform pair, we should concentrate on its restriction to the dual 
space $X^*$ of $X$.
$E$, as a function on $X^*$ is concave and upper semicontinuous 
from its definition, and general results of convex analysis 
(due to the known convexity and lower semicontinuity of $F$), guarantee that
\beq
F[\rho] = \sup
\left\{ E[v] - \langle v,\rho\rangle :\,  {v\in X^*} \right\}.
\label{inverse-LF-transform}
\eeq
There is an apparent asymmetry here due to our choices of signs,
but also a more glaring asymmetry: $L^1({\Bbb R}^3)$ is {\em not\/} 
the topological dual space of $L^\infty({\Bbb R})$.
This is repaired by using the weak and weak-$*$ topologies on $X$ and $X^*$
(see Appendix \ref{fa-appendix}) under which they are topological duals of each other.  
Symmetry is thus restored.  
Just as important, $F$ and $E$ are lower and
upper semicontinuous respectively with respect to these 
topologies due to the remarkable fact that for a convex
subset of a Banach space, closure with respect to the norm
or weak topology are the same.   

Lieb worked with $X=L^1\cap L^3$.
That has the advantage over $L^1$ that 
the Coulomb potential of a point charge is in the dual 
space $X^* = L^{3/2}+L^\infty$, but it seems we need to 
go outside $X^*$ to find all the potentials which may be of
interest anyway.
Consider, for example, the harmonic oscillator 
potential $|{\bm x}|^2$.  It is not in any of our $X^*$s,
yet it does not really pose any particular difficulty.
Eq. (\ref{LF-transform}) works for it, as long as we restrict 
the minimization to $\rho \in {\mathcal J}_N$.  [Outside that set, 
the right-hand side of Eq. (\ref{LF-transform}) involves nonsense 
such as $\infty + (-\infty)$.]
It begins to appear already that $V$-representability by a
potential in $X^*$ has a special status.  We will call this
restricted notion {\it $X^*$-representability.}

If we suppose that $\rho$ is {\it V-representable}, that is, a ground
state density for some potential $v[\rho]$, and that 
$F[\cdot\,] + \langle v,\cdot\,\rangle$ is somehow smooth,
Eq. (\ref{LF-transform}) suggests that the representing potential is 
some kind of functional derivative of $F$: $\delta F/\delta \rho = -v[\rho]$.
Smoothness is not in the cards, but convex analysis offers the derivative-like 
notion of {\it subdifferential\/} that does not depend on it.
The subdifferential of $F$ at $\rho$, traditionally denoted $\partial F$,
is a subset of $X^*$, the elements of which are called
{\it subgradients}.  The definition is
\[
v \in \partial F[\rho]
\iff
F[\rho^\prime] \ge F[\rho] + \langle v,\rho^\prime - \rho\rangle,\; \forall \rho^\prime\in X.
\]
A subgradient is the `slope' of a continous hyperplane which touches
the graph of $F$ at $\rho$, but is nowhere above it.
This is a looser notion than the ordinary one of tangency:
the subdifferential of an ice-cream cone at its point has many elements.
The relation (\ref{LF-transform}) shows that $\partial F(\rho)$ 
consists precisely of those potentials for which $\rho$ is a 
ground-state density.  From the Hohenberg-Kohn theorem, 
we deduce that the subdifferential at $\rho$ either is empty, or
its elements differ only by a constant.

\begin{figure}
\includegraphics[width=70mm]{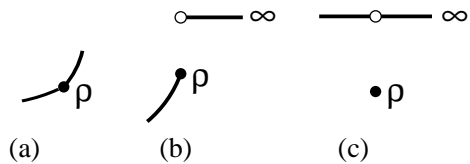}
\caption{Varieties of local behavior of $F[\rho]$ along linear slices in ${X}_N$,
through a density $\rho$.  In (b) and (c), $F$ takes the value $+\infty$.
}
\label{F-on-slices}
\end{figure}

How does the subdifferential relate to other notions of derivative?
Probably the simplest such is that of directional derivative.
The directional derivative of $F$ at $\rho \in {\dom} F$ in the 
direction $\delta \rho \in {X}$ is
\beq
F^\prime[\rho;\delta \rho] = 
\lim_{s\downarrow 0} \frac{1}{s}\Big( F[\rho + s\, \delta \rho] - F[\rho]\Big),
\label{directional-derivative}
\eeq
if the limit exists.  But, convexity guarantees that the difference
quotients are nonincreasing as $s\to 0$, so the limit {\em always\/}
exists, though it may be infinite.
In particular, if $F[\rho + s\, \delta \rho] = +\infty$ for all $s > 0$, 
then $F^\prime[\rho;\delta \rho] = +\infty$.
The domain of the directional derivative at $\rho$, denoted $\dom F^\prime[\rho;\cdot\,]$,
is defined as all $\delta \rho$ such that $F^\prime[\rho;\delta \rho] < +\infty$
If, for every $\delta \rho$, the directional derivative is finite and satisfies
$F^\prime[\rho;-\delta \rho] = -F^\prime[\rho;\delta \rho]$,
then the various directional derivatives fit together into a linear functional.  
If this linear functional is continuous, it is called a {\it G\^ateaux derivative},
or {\it G-derivative.} 
(Beware, the continuity requirement is not always imposed\cite{Phelps88}.) 
If $\rho$ is $X^*$-representable, the variational principle assures
us that 
\[
F^\prime[\rho;\delta \rho] \ge -\int v \, \delta \rho \, d{\bm x}.
\]
An important question then is to what extent equality holds.
There are certainly $\delta\rho$ for which $F^\prime[\rho;\cdot\,]$ 
is infinite.

The situation is illuminated by looking at varieties of local behavior of $F$ 
along line segments in $X$.  They are depicted in Fig. \ref{F-on-slices}.  
First, of course if $\int \delta\rho \, d{\bm x} \not= 0$, then
for any $s \not= 0$,
$\rho + s\, \delta\rho$ is outside $X_N$ and $F[\rho + s\, \delta\rho] = +\infty$,
so this is a trivial case of Fig. \ref{F-on-slices} (c).  
But, even if the perturbation does not stray from $X_N$, (b) and (c) occur for
every $\rho \in \dom F$.  This may be surprising at first.
But it is not so hard to find two sources of this kind of behavior.
First, there are perturbations for which
$\rho + s\, \delta\rho$, for $s$ on one or both sides of zero,
falls outside $X_N^+$, the set of elements of $X_N$ which are
everywhere non-negative.  For example, this will happen for $s > 0$ if 
the negative part of $\delta\rho$ falls off too slowly at infinity,
depending on $\rho$.  If the positive part also falls off too slowly,
then we have a picture like (c).
Secondly, there are perturbations for which $\rho + s\, \delta\rho$ is
in $X_N^+$ for $s$ in some open neighborhood of zero, but not in $\JN$.  
This, too, leads to case (c).
A severe enough discontinuity of $\delta\rho$, for example, will do the trick.
A variation on the construction we used to show that $F$ is locally
unbounded on $\JN$ provides a more subtle example.
Adding oscillations of the right amplitude to $\rho$ at a hierarchy of 
ever shorter wavelengths will produce a sequence converging with 
respect to $X$-norm to $\rho + \delta \rho \in X_N^+$, but with an 
{\em infinite\/} kinetic energy.

Falling outside $X_N^+$ under perturbation produces both scenarios
(b) and (c).  Falling outside $\JN$ while remaining in $X_N^+$ gives
only scenario (c).
To see this, consider $\rho, \rho^\prime \in \JN$, and suppose 
$\| \rho^\prime/\rho\|_\infty = M > 1$ is finite. Then, 
we already know that
$\rho_\lambda = (1-\lambda) \rho + \lambda \rho^\prime$
is in $\JN$ for $0 \le \lambda \le 1$. 
And, $\rho_\lambda$ is in $X_N^+$ for $-M \le \lambda < 0$.
For $\lambda$ in this latter range, $\rho_\lambda$ is in $\JN$
if $|\sqrt{\rho_\lambda}|_{H^1}$ is finite.
Since
\[
\frac{|\nabla\rho_\lambda|^2}{4\rho_\lambda}
 = \frac{1}{4\rho_\lambda}\Big( (1-\lambda)^2 |\nabla\rho|^2 + \lambda^2 |\nabla \rho^\prime|^2
+ 2\lambda(1-\lambda) \nabla\rho \cdot \nabla\rho^\prime \Big),
\]
and
$|2\lambda(1-\lambda) \nabla\rho \cdot \nabla\rho^\prime| \le
|\lambda(1-\lambda)| (|\nabla\rho|^2 + |\nabla\rho^\prime|^2)$,
taking $\lambda = - 1/(2M)$ 
gives $\tilde{\rho} \defeq \rho_{-1/2M} \ge \rho/2 \ge \rho^\prime/(2M)$, so that
\[
\frac{|\nabla\tilde{\rho}|^2}{4\tilde{\rho}}
 \le  6\frac{|\nabla\rho|^2}{4\rho}
 +  2M \frac{|\nabla\rho^\prime|^2}{4\rho^\prime} < \infty.
\]
But, $-1/(2M)$ is halfway to the critical value of $\lambda$,
and nothing prevents taking that step again and again.
Thus, if $\rho$ and $\rho^\prime$ are in $\JN$, the interior
of the line segment 
$\{ (1-\lambda)\rho + \lambda\rho^\prime : \lambda \in {\Bbb R}\} \cap X_N^+$
is also in $\JN$.  This argument does not tell us what happens at an end-point of 
this segment.  However, lower semicontinuity of $F$ implies that either it is 
in $\JN$, or $F$ diverges to $+\infty$ there.

This second class of density perturbations is apparently another 
short-distance problem, and will be absent in the coarse-grained model.
But, the first kind, connected with the fact that every density has a
tail which dies away, will still occur there.
On a slice entirely in $\dom F$, we will get something like
Fig. \ref{F-on-slices} (a).  We would hope that there was no kink, but
there is no proof that is the case, even if $F$ has a 
(necessarily essentially unique) subgradient at $\rho$. 

Here is a tempting argument put forward some time ago\cite{EE84}, 
which unfortunately keeps recurring.  It purports to show that at 
an $X^*$-representable density, a conventional G\^ateaux 
functional derivative exists and coincides with the representing potential.
The epigraph of $F$, the set of points on or above its graph,
\[
\epi F = \left\{ (\rho,z)\in X \times{\Bbb R} \, : \, F[\rho] \le z \right\},
\]
is convex.  Any directional derivative gives a line disjoint from 
the interior of $\epi F$, hence can be extended to a full hyperplane
by the Hahn-Banach theorem (the interior is the set of points which have
a full open neighborhood contained in $\epi F$).  Since such a hyperplane 
must coincide with the subgradient by assumption, the conclusion follows.
This {\em does not work}.  We have just seen that with Fig. \ref{F-on-slices}.
A G\^ateaux derivative does not give infinite directional derivatives.
That same figure also shows the flaw in the argument: 
{\em the interior of $\epi F$ relative to $X_N$ is empty}.

Of course, the only directional derivatives $F[\rho;\delta\rho]$
which have a real physical significance are those for which both
$\rho$ and $\rho + s\, \delta\rho$ for some $s > 0$ are
in $\dom F$.  The others are an artifact of the Banach space
imbedding, and are sure to be ill-behaved, so that
the best that could be expected is `quasi-G-differentiability,' 
meaning that all those physically significant directional
derivatives fit together into a continuous linear functional.

Now we turn to the $V$-representability issue.  Which densities are
ground state densities of some potential ($V$-representable)?
A coherent answer to the question probably requires deciding what will count as a potential.
One possibility, motivated by the Legendre-Fenchel duality, is to 
consider only $v$'s belonging to $X^*$.  This $X^*$-representability,
already introduced, is a somewhat restricted notion but it is well-structured 
and something can be said about it based on corollories of the celebrated Ekeland variational 
principle\cite{Ekeland79}.  These results will be used in Section \ref{continuum}.
From Eq. (\ref{LF-transform}) we have 
\[
E[v] \le F[\rho] + \langle v,\rho\rangle, \quad v\in X^*, \rho\in X.
\]
If $\rho$ and $v$ saturate the inequality, then $v \in \partial F(\rho)$.
Suppose that it is almost saturated, so
\beq
F[\rho] + \langle v,\rho\rangle \le E[v] + \epsilon.
\label{eps-subdiff}
\eeq
In that case, we say that $v$ is in the $\epsilon$-subdifferential 
$\partial_\epsilon F[\rho]$.  Physically, this means that $\rho$ 
almost achieves the ground-state energy in $v$.
Does that imply that it is almost a ground-state density of $v$?
Not quite, but the Ekeland variational principle does warrant 
(for example, Thm. I.6.2 in Ref. \cite{ET}) the following:
For any $\lambda > 0$, there is a density-potential pair 
$\rho_\lambda \in \dom F$ and  $v_\lambda \in X^*$, such that 
$\rho_\lambda$ is a ground-state density of $v_\lambda$ and
\beqa
F[\rho_\lambda] + \langle v,\rho_\lambda \rangle & \le &
F[\rho] + \langle v,\rho\rangle,
\nonumber \\
\|\rho_\lambda - \rho\|_{X}  & \le & \lambda,
\nonumber \\
\|v_\lambda - v\|_{X^*} & \le & \epsilon/\lambda.
\label{Ekeland-variational}
\eeqa
(Note that $v$ is on both sides of the first inequality.)
Interesting conclusions can be drawn from this.
Given a density $\rho$, pick any potential and find
$\epsilon$ to satisfy inequality (\ref{eps-subdiff}).  Then, taking
$\lambda = 1/n$, the theorem gives us a sequence of $X^*$-representable
densities converging to $\rho$, showing that $X^*$-representable densities
are dense in $\dom F$.  This is the Br\o nsted-Rockafellar theorem.
But, note that control over the norms of the representing potentials 
gets progressively worse as $n\to\infty$.
Another interesting choice is to set $\lambda = \sqrt{\epsilon}$.
In that case, 
$\|\rho_\lambda - \rho\|_{X}  \le \sqrt{\epsilon}$ and
$\|v_\lambda - v\|_{X^*}  \le \sqrt{\epsilon}$,
so that if $\epsilon$ is small only small perturbations of both 
$\rho$ and $v$ to find a ground-state density/representing potential pair.
This does {\em not\/} say that all nearby $X^*$-representable densities
have nearly the same representing potential.

In thinking about the relationship between $X^*$-representable densities and their 
associated potentials, it is useful to work in the product space $X \times X^*$.
The graph of the subdifferential, 
\[
\graph \partial F \defeq \left\{ (\rho,v)\in X \times X^* \, : \, v \in \partial F(\rho)\right\}, 
\]
then naturally suggests itself as an object of study.
If a pair $(\rho,v)$ is not in $\graph \partial F$, then $\rho$ is not
a ground state density for $v$.  This is a little crude.  It might be
even better to have a function which measured by how much it fails.
To this end, we introduce the {\it energetic excess} of $\rho$ with respect to $v$,
\beq
\Delta[\rho,v] \defeq F[\rho] + \langle v, \rho \rangle - E[v].
\eeq
$\Delta[\rho,v]$ is the lowest energy attainable with quasi-density $\rho$ in the
presence of $v$, relative to the ground state energy of $v$, and therefore
$\Delta: X\times X^* \to [0,\infty]$.  Clearly, $\graph \partial F = \Delta^{-1}(0)$.
Note that, although the natural description of the significance of $\Delta$ seems
asymmetric, $\rho$ and $v$ are really on essentially equivalent footing in
$\Delta[\rho,v]$ as follows from the conjugacy of $F$ and $E$.  

Some properties of $\Delta$ follow immediately from those of $F$ and $E$
which have already been discussed.
Namely, it is convex separately in each of $\rho$ and $v$ since the pairing 
$\langle v,\rho\rangle$ is linear, and it is lower semicontinuous with respect to 
either (weak)$\times$(norm) or (norm)$\times$(weak-*) topology on $X\times X^*$  
since $(v,\rho) \mapsto \langle v,\rho\rangle$ 
is continuous with respect to either of these.
This means in particular that $\graph \partial F = \Delta^{-1}(0)$ is closed 
with respect to either of these topologies.
In fact, if only
\[
\Delta[\rho_n,v_n] \to 0 \quad \mbox{and}\quad (\rho_n,v_n) \to (\rho,v),
\] 
then 
\[
(\rho,v)\in \Delta^{-1}(0),\;\; F[\rho] = \lim_{n\to\infty} F[\rho_n],\;\;
\mbox{and}\;\; E[v] = \lim_{n\to\infty} E[v_n].\]
The last two limits follow since $F$ can only decrease at the limit and $E$ 
can only increase, but $\Delta \ge 0$.

The bonus convergence of $F$ and $E$ implies that the subset of 
$\graph \partial F$ which respects the gauge convention of Eq. (\ref{gauge-convention})
is itself closed.  And, it further follows that the map $\rho \mapsto \|v[\rho]\|_{X^*}$,
like $F[\cdot]$, is lower semicontinuous, where for convenience we declare
$\|v[\rho]\|_{X^*} = +\infty$ if $\rho$ is not $X^*$-representable.  
For, suppose that $v[\rho]$ exists but $\|v[\cdot]\|_{X^*}$ is not lower semicontinuous
at $\rho$.  Then, there is a sequence $\rho_1,\rho_2,\ldots$ in the domain of
$v[\cdot]$ such that $\|v[\rho_j]\|_{X^*} \le \|v[\rho]\|_{X^*} - \epsilon$.
Since $X^*$ is a Banach space dual, bounded subsets are weak-* compact, which
means that there is a subsequence $v[\rho_{n_j}]$ which is weak-* convergent
to $v_0$ with norm not exceeding $\|v[\rho]\|_{X^*} - \epsilon$.  
But that would mean that $v_0$ is {\em the} representing potential for $\rho$
respecting the gauge-convention, which is clearly impossible since its norm
is too small.

Part of the conclusion drawn from Ekeland's variational 
principle [see (\ref{Ekeland-variational})] takes a curious form when 
written in terms of the energetic excess $\Delta$.  
We equip $X\times X^*$ with the norm $\|(x,w)\|_{X\times X^*} = \|x\|_X + \|w\|_{X^*}$.  
Since this is stronger than the topologies discussed in the previous paragraph, the
convergence results stated there still hold.  Now, if $\Delta[\rho,v] \le \epsilon$,
then there is $(\rho^\prime,v^\prime) \in \Delta^{-1}(0)$ with  
{$\|(\rho^\prime,v^\prime) - (\rho,v)\|_{X\times X^*} \le 2\sqrt{\epsilon}$}.
Thus, the minimum of $\Delta$ is necessarily ``narrow'' in some sense.

Beyond $X^*$-representability, not much can be said in general because not
much is known.  Englisch and Englisch\cite{EE83} produced some examples 
of densities (in $\JN$) which are not $V$-nonrepresentable by any potential 
which is a function.  This shows at least that such densities exist.
The examples are for a single particle and use the fact that the only candidate 
potential in that case is given by $(\nabla^2\sqrt{\rho})/\sqrt{\rho}$.
These examples were further analysed by Chayes, Chayes and Ruskai\cite{CCR85}.  
One is $\rho = (a + b |x|^{\alpha + 1/2})^2$ for $x$ near zero with 
$0 < b < a$ and $0 < \alpha < 1/2$.  
Non-V-representable densities can be undeniably physical, as shown by
the examples\cite{Martin} of single-particle excited states in a central
potential with a node in the radial wavefunction, $\rho \sim (r-r_0)^2$.
The non-V-representability of all of these examples is clearly due to
short-distance problems, which does not prove lack of other sorts, but
provides more motivation for short-distance regularization.  

Thus, $X^*$-nonrepresentable and even V-nonrepresentable exist, but
how common are they?  Suppose we equip $\JN$ with the metric
$d(\rho,\rho^\prime) = \|\rho-\rho^\prime\|_{X} + |F[\rho]-F[\rho^\prime]|$.
Prop. \ref{bad-potentials} in Appendix \ref{bad-potentials} shows that
$\rho \mapsto \|v[\rho]\|_{X^*}$ is nowhere upper semicontinuous,
and this has the consequence that $X^*$-nonrepresentable 
densities are {\it topologically generic} in $\JN$ with respect to the
topology induced by the metric $d$.  In this sense, most densities
in $\JN$ fail to be $X^*$-representable.  
The demonstration of this involves potentials which oscillate on extremely
short wavelengths, so this is again a short-distance issue.  However,
one should note that only low intrinsic-energy densities are involved
(at least explicitly).  

It might be suggested that absence of V-representability
over all of $\dom F$ is disquieting, but far from a disaster.
Ultimately, we are only interested in the V-representable
densities, whichever ones those might turn out to be.
The problem becomes more acute in Kohn-Sham theory.
The Kohn-Sham decomposition of the intrinsic energy is 
\beq
F[\rho] = T_s[\rho]+E_H[\rho]+E_{xc}[\rho].
\label{KS-decomp}
\eeq
Here, $T_s[\rho]$ denotes the intrinsic energy for a non-interacting system 
of quasi-density $\rho$.  Note that this has the same domain as does $F$.
The Hartree energy $E_H[\rho]$ is simply the classical Coulomb energy of 
the charge distribution $-e \rho$.  The exchange-correlation energy 
$E_{xc}[\rho]$ is then defined by the equation.
The point of this decomposition is that it facilitates a sort of
self-consistent field approach.  One solves a problem for non-interacting
particles with a potential which is the sum of external, Hartree and
exchange-correlation potentials, self-consistency being imposed on the
latter two.  This has proved a very successful strategy for practical computations.
If $\rho$ is a ground-state density of an interacting system in
external potential $v[\rho]$, then it should be a ground state density of the
non-interacting system in external potential
\beq
v_s[\rho] = v[\rho] + \phi[\rho] + v_{xc}[\rho],
\label{KS-potls}
\eeq
with $\phi[\rho] = \delta E_H/\delta \rho$ and $v_{xc}[\rho] = \delta E_{xc}/\delta \rho$
(note the sign convention).
This raises a couple of problems.  What if $\rho$ is (interacting)
V-representable but not non-interacting V-representable?
Then, evidently, the solution does not exist.
Consequently, it is very important to the Kohn-Sham enterprise that
the sets of V-representable, or maybe we should say $X^*$-representable,
densities for the interacting and non-interacting systems coincide.
At a density in the intersection, Eq. (\ref{KS-potls}) can be used to define
$v_{xc}[\rho]$.  Even then, however, the existence of a functional derivative 
of $E_{xc}$ and its coincidence with $v_{xc}$ is not assured.  
Subdifferentiability of $F$ and $T_s$ is not strong enough for that.

Since it is the bad behaviors of the fine-grained theory which interest us,
let us summarize them.  $F$ is not continuous in $X$-norm, 
{\em even restricted to its effective domain}.  There are densities
that are not V-representable, and many more that are not $X^*$-representable.
This puts two blocks in the way of well-definedness of a functional
derivative of $F$.  First the non-universality of subdifferentiability,
and second the lack of a demonstration that directional derivatives
coincide with a subgradient.  All of these seem at least partially 
connected to abnormal occurences at asymptotically short distance scales.
This motivates the coarse-grained approach, wherein the short-distance
scale degrees of freedom are allowed to relax energetically.  

\section{Coarse-Graining: Notation and Background}
\label{CG}

In this section, we turn our attention to the coarse-grained model.
Much of it will formally resemble the fine-grained theory
outlined in the previous section.  
{\em The term `density' is intended in the coarse-grained sense
unless otherwise noted}.

As explained in the Introduction, a (coarse-grained) {\it quasi-density} is 
a collection of number observables $\rho({\bm R}) \mathrm{vol}({\bm R})$,
one for each cell ${\bm R}$ of a partition $\Part$ of ${\Bbb R}^3$, 
satisfying a summability condition (see below).  
For purposes of working in a vector space, the ``number observables''
are allowed to have any real value.
Thus, a quasi-density is identified with an equivalence class of elements 
of $X$, where two elements of $X$ are equivalent if they have equal
integrals over every cell.  Ultimately, the only members 
of such an equivalence class which are essential are those
which belong to $\JN$, if there are any.
Still, a useful and easily visualized surrogate for a quasi-density is the 
everywhere-defined function which is uniform over each cell 
${\bm R}$ (${\mathfrak P}$-measurable), and equal to 
$\rho({\bm R})$.  This is a special element of the equivalence
class associated with $\rho$ as we have defined it, but it far 
too discontinuous to be in $\JN$.
Some variation in the sizes and shapes of the cells of $\Part$ is permissible, 
but it must be limited (see Ref. \cite{Lammert06}).  The partition generated 
by any sort of regular lattice is acceptable, and for simplicity, we keep 
a cubical lattice in mind.  

For a real-valued function $f$ on $\Part$, define the $L^1$ norm 
$\|f \|_{1} = \sum_{{\bm R}\in{\Part}} |f({\bm R})| \mathrm{vol}({\bm R})
 = \int |{f}({\bm r})| \, d{\bm r}$.
The Banach space of such functions with finite norm is denoted $\X$.
This is the space of {\it quasi-densities}.
Certain subsets of ${\X}$ are given names.
The set of everywhere non-negative elements, that is, the
(proper) densities, is denoted $\XP$; that of those everywhere strictly 
greater than zero is denoted $\XPP$.
Subsets of elements which integrate to a particular value are indicated with 
a subscript, e.g., ${\XNP}$ denotes the properly normalized densities,
and legitimate density perturbations are located in ${\X}_0$.
It is the nature of the normalization condition in particular which suggests
the $L^1$ norm.  Although it seems fairly natural, one might question its
appropriateness.  Most of the results actually depend only on the topology
associated with this norm, and one should note that it is the weakest
topology making the number-in-cell observables and the total number
continuous.  More discussion of this point can be found in Section \ref{continuum-ideas}.

The intrinsic energy is defined as in Eq. (\ref{F1}).  
Since the equivalence class in $X$ associated to $\rho \in {\X}$ 
intersects $\JN$ if and only if $\rho \in \XPP$,
the effective domain of $F$ is
\[
\dom F \defeq \{\rho\in\X : F[\rho] < +\infty\} = \XNP.
\]
The infimum in the definition of $F[\rho]$ is guaranteed to be attained\cite{Lammert06}, 
but there is no guarantee that there is only one {\em fine-grained} density which attains it.
In some rough sense, $\dom F$ here is a much bigger subset of $\X$ than 
$\JN$ is of $X$.
As in the fine-grained theory, $F$ is convex and lower semicontinuous\cite{Lammert06}. 
One thing coarse-graining has achieved is a bound for $F$:
\beq
0 \le F[\rho] \le N \Fmax, \quad \rho \in{\XNP}.
\label{F-bound}
\eeq
$N \Fmax$ is the energy required to pack all $N$ particles into a 
single cell (maximum such if the cells are not identical).

The main result needed from Ref. \cite{Lammert06} concerns {\em unique and universal} 
V-representability.  As discussed in the Introduction, potentials are constant on cells.  
${\X}^*$, the dual space of $\X$, consists of $\Part$-measurable functions which are
uniformly bounded.  But we must go outside that space to find many of the needed potentials.
The interesting set, that of $\Part$-measurable functions bounded below but not necessarily above,
is denoted ${\V}$. 

\begin{thm}[Coarse-grained Hohenberg-Kohn]
For $\rho \in {\X}^{++}$, the infimum in the definition of $F[\rho]$
is attained (at mixed state $\gamma$, say) and there is precisely one potential 
$v[\rho] \in {\V}$ with a ground state ($\gamma$) 
of density $\rho$.
\label{CG-HK-P}
\end{thm}
\begin{proof}
See Ref. \cite{Lammert06}.
\end{proof}
The constant offset in the potential is fixed here by the 
convention of Eq. (\ref{gauge-convention}):
\beq
\langle T + V_{ee} + v \rangle_\gamma = E[v[\rho]] = 0.  
\label{offset-convention}
\eeq
This has the effect that all external potentials associated 
with densities are bounded below by $-\Fmax$ (see Eq. \ref{F-bound}).

\section{Single-Scale Model}
\label{single-scale}

In this section, the new results on the single-scale coarse-grained
theory are studied.  
In Section \ref{single-scale-ideas}, the Theorems are stated,
and discussed with some indication of the methods of proof.
Detailed proofs are given in Section \ref{ss-proofs}.

\subsection{ideas}
\label{single-scale-ideas}

It is cumbersome to have to maintain strict normalization
of densities and states at all times, so we will work with a slight 
modification of $F$. $\Fhat[\rho]$ is defined as in Eq. (\ref{F1}),
but with the normalization restriction on states lifted.  
Thus, $\Fhat$ agrees with $F$ on $\XNP$ (which is the physically 
important set).  But, for $\rho\in\XP$, it scales linearly with the normalization, 
$\Fhat[\lambda\rho] = \lambda \Fhat[\rho]$, and is $+\infty$ elsewhere.  
Note this maintains convexity and lower semicontinuity.  
The theorems are stated in terms of $\Fhat$, but their translations into
terms of $F$ are easy.

As seen already, coarse-graining renders $F$ bounded in a fairly trivial way.
It is actually even continous with respect to the $L^1$ topology on its 
domain, ${\X}^{+}$.  
\begin{thm}
\label{Continuity-Thm}
$\Fhat$ is continuous on ${\X}^{+}$ with respect to the $L^1$ topology.
\end{thm}
Since $F[\rho]$ is already known to be lower semicontinuous, proving this
requires showing only upper semicontinuity: $F$ does not exceed 
$F[\rho] + \epsilon$ in some neighborhood of $\rho$. 
The proof involves showing that wavefunctions
can be deformed to produce nearby densities without costing too much energy.
Two aspects of the coarse-grained situation make this possible.
First, imperfections are hidden.  We have to get only the cell averages right.
Secondly, $\X$ is, so to speak, locally finite dimensional.  Thus, to show
that any required small modification can be made in a bounded region $\Sigma$
only requires showing that a finite number of directions of modification
can be handled.  If we had an infinite number of dimensions, as in the fine-grained
theory, the margin for change in successive directions can shrink to zero.
In fact, there are infinitely many dimensions --- outside $\Sigma$.
These are handled in a completely different way: the tail of the state
can be chopped off, at arbitrarily small energy cost far enough out.
Then, using the fact that we have mixed states to work with, a new
tail can be grafted on with the energy cost bounded by Eq. (\ref{F-bound}).
(For arbitrarily normalized $\gamma$, that bound extends to 
$\langle T + V_{ee}\rangle_\gamma \le \Fmax \|\rho\|_1$, where $\gamma \mapsto \rho$.)
This is the only theorem which requires getting behind the densities and working
with states, but it is a key ingredient in Theorem \ref{Diff-Thm}.

Theorem \ref{CG-HK-P} showed that any $\rho \in {\X}_N^{++}$ is V-representable.  
Naively, we expect $-v[\rho]$ to coincide with the functional derivative 
$\delta F/\delta \rho$.  If $-v[\rho]$ is in ${\X}^*$ (bounded), it is a 
subgradient of $F$ at $\rho$.  But as discussed in Section \ref{fine-grained}, 
it is unclear to what extent directional derivatives $F[\rho;\delta\rho]$
agree with it even in that case.  
Certainly, it is impossible to put {\em all} directional derivatives together 
into a {\em linear} functional, so that classical G\^ateaux differentiability 
is out of the question, even if $v[\rho] \in {\X}^*$.
Situations (b) and (c) of Fig. \ref{F-on-slices} still arise for directions
$\delta\rho$ which cause immediate exit from $\XNP$.
But, since those are the only directions which cause that problem, 
it is from the beginning less severe than for the fine-grained interpretation.
Fortunately, directions $\delta\rho$ which lead immediately out of ${\X}_N^+$
do not seem to hold any physical interest.
It is certainly satisfactory if
$\langle -v[\rho],\delta\rho \rangle  = F^\prime[\rho;\delta\rho]$ for
$\delta\rho$ satisfying $\rho + s\, \delta\rho \in \XNP$ for some $s > 0$.  
That this set of directions really is coextensive with
\[
{\dom} F^\prime[\rho;\cdot\,] \defeq \{ \delta\rho \in {\X}: \, F^\prime[\rho;\delta\rho] < +\infty \},
\]
follows from
\[
F[\rho + s\, \delta\rho] \ge F[\rho] + sF^\prime[\rho;\delta\rho],
\]
which is a consequence of convexity.
It is not ruled out that $F^\prime[\rho;\delta\rho] = -\infty$ 
for some $\delta\rho \in \dom F^\prime[\rho;\cdot\,]$.
However, this could happen only if $\rho + s\, \delta\rho$
falls outside $\XNP$ for all $s < 0$, because
$F^\prime[\rho; -\delta\rho] \ge -F^\prime[\rho; \delta\rho]$,
which is another simple consequence of convexity.

\begin{thm}
\label{Diff-Thm}
Suppose $\rho \in {\mathcal X}^{++}$ is represented by the potential $v$.
For $\delta \rho \in {\X}$, either
\renewcommand{\theenumi}{\alph{enumi}}
\begin{enumerate}
\item there is no $s>0$ for which $\rho+s\delta \rho \in {\X}^+$,
in which case $\Fhat^\prime[\rho;\delta \rho] = +\infty$, or
\item
$-\infty \le \Fhat^\prime[\rho;\delta \rho] < +\infty$ and 
$\Fhat^\prime[\rho;\delta \rho] = -\int v \, \delta \rho \, d{\bm x}.$
\end{enumerate}
\end{thm}

To paraphrase, if $\rho \in {\XNPP}$ is represented by 
the potential $v$, then for any $\delta \rho \in {\X}$ we have the 
following dichotomy:
Either $\rho+s\, \delta\rho \notin {\XNP}$ for any $s>0$, so that
$\delta\rho$ is simply not in $\dom F^\prime[\rho;\cdot]$,
and $F^\prime[\rho;\delta \rho] = +\infty$, 
or {$F^\prime[\rho;\delta \rho] = -\int v \, \delta \rho \, d{\bm r}$}.

There are domain issues which should be discussed before
sketching the method of proof.  
In the fine-grained case, only
potentials in $X^*$ could be handled systematically, and these
are bounded by nature: 
$\left| \int -v \, \delta\rho \, d{\bm x} \right| \le \|v\|_{X^*} \|\delta\rho\|_X$.
Now, $v[\rho]$ lies in ${\V}$, and it is not immediately
obvious when $\int -v[\rho] \, \delta\rho \, d{\bm x}$ is well-defined, since
both the positive and negative parts can be infinite for some $\delta\rho\in \X$.  
But in fact, $\int v[\rho] \cdot \, d{\bm x}$ is unambiguous on all
of $\dom F^\prime[\rho;\cdot\,]$, as the following argument shows.  
Certainly, $v[\rho]\rho$ is integrable.
Shift $v[\rho]$ by a (finite) constant to make it positive and
split $\delta\rho$ into positive and negative parts as
$\delta\rho = (\delta \rho)^+ - (\delta \rho)^-$,
so that
$-v[\rho]\, \delta \rho  = -v[\rho](\delta \rho)^+  + v[\rho](\delta \rho)^-$.  
If $\delta \rho \in {\dom} F^\prime[\rho;\cdot\,]$, then $(\delta \rho)^-/\rho$ is bounded,
and since $v[\rho] \rho$ is integrable, $v[\rho](\delta \rho)^-$ is also.
Thus, for such $\delta \rho$, $\int -v[\rho]\, \delta \rho \, d{\bm r}$ is either 
real, or $-\infty$.  The second possibility is not particularly exotic if $v[\rho]$
diverges as $r\to\infty$.

The argument used to prove the theorem is similar to the one
suggested in Section \ref{fine-grained} for G\^ateaux differentiability
of $F$ at V-representable densities and which was shown to fail in that
context.   
It follows from convexity of $F$ that $F^\prime[\rho;\cdot\,]$ 
satisfies $F^\prime[\rho;x+y]  \le F^\prime[\rho;x] + F^\prime[\rho;y]$,
and $F^\prime[\rho;\lambda x]  = \lambda F^\prime[\rho;x]$ for $\lambda \ge 0$.
That is, it is {\it sublinear}, but not necessarily continuous.
Since $\int v[\rho] \cdot \, d{\bm x}$ is finite and linear on
density perturbations with {\em bounded} support (call this set ${\mathcal K}$), 
the sum functional $F^\prime[\rho;\cdot\,] + \int v[\rho] \cdot \, d{\bm x}$ is also
sublinear on such perturbations.  Now, the variational principle guarantees that
this functional is non-negative, so the problem is to show that it is
exactly zero on $\dom F^\prime[\rho;\cdot\,]$.  If there is a direction
in ${\mathcal K}$ along which it is not, then on that one-dimensional
space, there is a nonzero linear functional dominated by
$F^\prime[\rho;\cdot\,] + \int v[\rho] \cdot \, d{\bm x}$.
One version of the Hahn-Banach theorem says that such a linear functional 
can be extended to a linear functional $w$ on the entire vector space ${\mathcal K}$
which is still dominated by $F^\prime[\rho;\cdot\,] + \int v[\rho] \cdot \, d{\bm x}$,
since the latter is sublinear.  
In the coarse-grained setting, $w$ is guaranteed
to have the same form as a potential (in the fine-grained case, distributions,
among other things, might arise as linear functionals), and that would make
$v+w$ another potential having $\rho$ as ground state density, violating
the Hohenberg-Kohn theorem.  The result is then extended to all of
$\dom F^\prime[\rho;\cdot\,]$ using Theorem \ref{Continuity-Thm} and
convexity of $F$.

Since all densities are ${\V}$-representable in the
coarse-grained model, a natural next question is whether $v[\rho]$
is a continous function of $\rho$.
This would say that, if $\|\rho^\prime - \rho\|_1$ is small enough,
$v[\rho^\prime]$ is `close' to $v[\rho]$, which would seem to require 
a topology on ${\V}$ to make `close' meaningful.
One topology to consider is the product topology, in which a
neighborhood of $v$ consists of all $v^\prime \in {\V}$
which are close to $v$ on a specified bounded region, but unconstrained
outside it, so that open sets are unions of sets of the form
\[
U(v,\Sigma,\epsilon) = \{v^\prime \in {\V}: 
|v^\prime({\bm x}) - v({\bm x})| < \epsilon, \; \forall {\bm x} \in \Sigma \},
\]
for bounded $\Sigma$.
This topology really is weak in the current context:
for example, if $v_n$ is zero for $|{\bm x}| < n$, but goes below
$-n$ and above $n$ somewhere outside that radius, it converges to
zero in the product topology.
Be that as it may, we can prove that $\rho \mapsto v[\rho]$ is
continuous from $\XNP$ with the $L^1$ topology (as always) to
${\V}$ with the product topology.
What needs to be shown is that the restriction $v[\rho]\rest_\Sigma$ 
of $v[\rho]$ to
a bounded region $\Sigma$, viewed simply as a vector in a finite-dimensional
Euclidean space, is continuous as a function of $\rho$.
The key is to view $F[\rho]$ as a family $F[\rho\rest_\Sigma; \rho\rest_{\Sigma^c}]$
of functions of the finite dimensional variable $\rho\rest_\Sigma$ 
parametrized by $\rho\rest_{\Sigma^c}$, where the superscript `$c$' indicates
a complement, i.e., $\Sigma^c = {\Bbb R}^3\setminus\Sigma$.  These are differentiable convex 
functions on a finite-dimensional space, continuous with respect to the
parameter $\rho\rest_{\Sigma^c}$, and the very nice properties of
convex functions on finite-dimensional spaces imply that the derivatives,
which are $v[\rho]\rest_{\Sigma}$, are continuous in both 
$\rho\rest_{\Sigma}$ and $\rho\rest_{\Sigma^c}$.

A stronger statement is available, but not, as might be expected, by 
using a stronger topology on ${\V}$.  
The crucial observation is that, although $v[\rho]$ may be unbounded, 
the product $v[\rho]\rho$ is always relatively tame.  
It is integrable, that is, it has finite $L^1$ norm.
\begin{thm} 
\label{V-Cont-Thm}
The map $\rho \mapsto v[\rho] \rho$ of ${\XPP}$ into $\X$ is continuous 
with respect to $L^1$-norm.
\end{thm}
Theorem \ref{V-Cont-Thm} says that the map 
$\rho \mapsto v[\rho] \rho$ of $\XNPP$ into $\X$ is continuous with 
respect to $L^1$ norm.
Technically, this is different than saying that $\rho \mapsto v[\rho]$
is continuous, but it has a similar import.
The theorem represents a strengthening of the earlier
result because, for a bounded region $\Sigma$,
\[
\int_\Sigma \left| (v[\rho^\prime] - v[\rho])\right| \rho^\prime \, d{\bm x}
 \le
\int_\Sigma \left| v[\rho^\prime]\rho^\prime - v[\rho]\rho \right| \, d{\bm x}
+ \int_\Sigma \left| v[\rho] (\rho^\prime - \rho) \right| \, d{\bm x}.
\]
For $\rho^\prime$ in a small enough neighborhood of $\rho$ (depending on
$\rho$ indirectly through $v[\rho]$) the second integral on the right-hand
side can be made small and $\rho^\prime$ will be bounded uniformly away
from zero on $\Sigma$, so that some multiple of
$\int_\Sigma \left| v[\rho^\prime]- v[\rho] \right| \, d{\bm x}$
is bounded by 
$\int_\Sigma \left| v[\rho^\prime]\rho^\prime - v[\rho]\rho \right| \, d{\bm x}$
plus a small correction.

The Proposition in Appendix \ref{bad-potentials} shows that 
this quasi-continuity of $v[\rho]$ does not extend to even 
the $X^*$-representable fine-grained densities in $\JN$.

\subsection{deferred proofs}
\label{ss-proofs}

The next lemma is preparation for proving Thm. \ref{Continuity-Thm}.

\begin{lem}
\label{Deformation-Lemma}
Suppose $\gamma_0 \mapsto \rho_0$ is a ground state of $T + V_{ee} + v$,
and let $\Sigma$ be a bounded region consisting of entire cells.  
Then, there exists a neighborhood $\Omega$ of $\rho_0\rest_\Sigma$,
and a family of states $\gamma[\rho\rest_\Sigma]$ indexed by 
$\rho\rest_\Sigma \in \Omega$ such that:
$\gamma[\rho_0\rest_\Sigma] = \gamma_0$,
the density of $\gamma[\rho\rest_\Sigma]$ is $\rho\rest_\Sigma$ on $\Sigma$
and $\rho_0\rest_{\Sigma^c}$ outside it,
and
$\Tr \{ (T+V_{ee})\gamma[\rho\rest_\Sigma] \}$ is an infinitely differentiable
function of $\rho\rest_\Sigma$.
\end{lem}

\noindent{\it Remark.}
Since $\Fhat[\rho]$ is lower semicontinuous, the
lemma implies that $\Fhat[\rho]$ varies continuously with 
density perturbations in a {\em spatially bounded} region.
It is a key step toward Theorem \ref{Continuity-Thm}, which
allows perturbations with unbounded support.

\begin{proof}
First, assume that $\gamma$ is a pure state $\Psi$.  The general case will follow 
very easily from that.  

Let $U_{\bm R}$ denote the region of $N$-particle configuration 
space where all particles are in the interior of a cell $\bm R$ in $\Sigma$.
$U_{\bm R}$ is open. 

Now, we appeal to the unique continuation principle which assures us that 
$\Psi$ is not identically zero on $U_{\bm R}$.  This is where we need $\Psi$ to
be an eigenstate of $T+V_{ee}+v$.
The simplest version of the principle, only requiring a locally bounded
potential suffices here.  See the Appendix to \S XIII.13 in Ref. \cite{RSIV}.  
Choose a smooth ($C^\infty$) non-negative function $g_{\bm R}({\bm x})$ on 
configuration space, compactly supported in $U_{\bm R}$, normalized so that 
$\int g_{\bm R}({\bm x}) |\Psi|^2 \, d{\bm x} = 1$, and define 
${\psi}_{s({\bm R})} = (1+s({\bm R}) g_{\bm R})^{1/2}\Psi$,
for a real parameter $s({\bm R})$.
Then, $1+s({\bm R}) g_{\bm R} > 0$ for $s({\bm R})$ sufficiently small,
say $|s({\bm R})| < \epsilon_{\bm R}$.
Using the Cauchy-Schwarz inequality, the finiteness of the various pieces of
the energy of $\Psi$, and the fact that derivatives of $g_{\bm R}$ are continuous
and compactly supported,  it is easy to see that $\|\nabla{\psi}_{s({\bm R})}\|^2$
and $\langle {\psi}_{s({\bm R})}|V_{ee}|{\psi}_{s({\bm R})}\rangle$
are smooth functions of $s({\bm R})$.
The (coarse-grained!) density of ${\psi}_{s({\bm R})}$ is identical to that of
$\Psi$ except in cell $\bm R$, where it is $\rho({\bm r}) + s({\bm R})$.
Now, to complete the construction, just repeat with the other cells in $\Sigma$, to get
${\psi}_{\bm s} = 
\prod_{{\bm R}\in\Sigma} (1+s({\bm R}) g_{\bm R})^{1/2}\, \Psi$.
We observe in passing that to get smoothness for variations in arbitrary directions
requires bounds on the derivatives uniform with respect to cell indices.  This is
a major part of the reason $\Sigma$ must be finite.

If $\gamma$ is a mixed state, we can perform the modification on just one of
the pure states in its canonical decomposition.  
\end{proof}

\noindent{\bf Proof of Theorem \ref{Continuity-Thm}.}
As $\Fhat$ is lower semicontinuous, we only 
need to show upper semicontinuity at $\rho$.  Further, density can always be {\em added}
with a density matrix corresponding to the desired extra density at an intrinsic 
energy cost bounded according to $\Fhat[\rho] \le \|\rho\|_1 \Fmax$,
the appropriate variant of Eq. (\ref{F-bound}).

So, we only need to show that, given $\epsilon$, there is $\delta > 0$ such that 
$\Fhat[\rho^\prime] < \Fhat[\rho]+\epsilon$ whenever $\rho^\prime$ is in 
the $\delta$-ball centered at $\rho$,
$B_\delta(\rho) \defeq \{ \rho^\prime \in {\X}^+\, : \|\rho^\prime-\rho\|_1 \le \delta \}$
and satisfies $\rho^\prime \le \rho$ everywhere.
We will refer to $B_\delta(\rho)\cap\{\rho^\prime: \rho^\prime \le \rho\}$ as the `lower half' 
of $B_\delta(\rho)$.

Suppose $\gamma = \sum \lambda_\alpha |\psi_\alpha\rangle\langle \psi_\alpha|$ 
is a ground state with density $\rho$.  We modify it as follows.
Let $\varphi(x)$ be a smooth, monotonically decreasing, function 
${\mathbb R}^+ \to {\mathbb R}^+$ which is $1$ for $0 \le x \le 1$ and $0$ for $x \ge 2$,
and define 
$\Phi_L {\defeq} \prod_{i=1}^N \varphi(|{\bm x}_i|/L)$.
The modified state is 
\[
\gamma_L = \sum \lambda_\alpha | \Phi_L \psi_\alpha\rangle \langle \Phi_L \psi_\alpha|.
\]
Then $\gamma_L \mapsto \rho_L$, where $\rho_L = 0$ outside a sphere $S_{2L}$ of radius $2L$.

As $L \to \infty$, $\rho_L$ converges to $\rho$ uniformly on any given bounded region
$\Sigma$ and the intrinsic energy of $\gamma_L$ tends to that of $\gamma$.
Find $L$ large enough that
\[
\|\rho-\rho_L\|_1 < {\epsilon}/({3\Fmax}),
\]
and
\[
\Fhat[\rho_L] \le \langle T+V_{ee}\rangle_{\gamma_L} < \Fhat[\rho]+\epsilon/3.
\]
By  Lemma \ref{Deformation-Lemma}, there is some $\delta$ such that,
for $\tilde{\rho}$ in $B_\delta(\rho_L)$ and supported in $S_{2L}$,
\[
\Fhat[\tilde{\rho}] < \Fhat[\rho_L] + \epsilon/3 < \Fhat[\rho]+ 2\epsilon/3.
\]

Now, take $\rho^\prime$ in the lower half of $B_\delta(\rho)$,
and decompose it as $\rho^\prime = {\min}(\rho^\prime,\rho_L) + \rho^{\prime\prime}$,
so that $\rho^{\prime\prime} \le \rho-\rho_L$.
Then, according to the previous paragraph,
\[
\Fhat[{\min}(\rho^\prime,\rho_L)] < \Fhat[\rho] + 2\epsilon/3,
\]
and
$\rho^{\prime\prime}$ can be added with intrinsic energy cost not 
exceeding $\epsilon/3$.  Thus, $\Fhat[\rho^\prime] < \Fhat[\rho] + \epsilon$.
\qedsymbol

The next lemma is preparation for Thm. \ref{Diff-Thm} on directional 
derivatives of $\Fhat$, and deals with the special case of perturbations
with spatially bounded support.
The subset of ${\X}$ consisting of {\it simple functions}, functions which
are nonzero only on a bounded set, is denoted ${\mathcal K}$.

\begin{lem}[] 
\label{Baby-Diff}
If $\rho \in {\X}^{++}$ 
is represented by the potential $v$ and $\delta \rho \in {\mathcal K}$,
then 
\[
\Fhat^\prime[\rho;\delta \rho] = -\int v \, \delta \rho \, d{\bm r}.
\]
\end{lem}

\begin{proof}
It follows from convexity of $\Fhat$ that $\Fhat^\prime[\rho;\cdot\,]$ is a sublinear 
functional on ${\mathcal K}$, i.e., 
$\Fhat^\prime[\rho;x+y]  \le \Fhat^\prime[\rho;x] + \Fhat^\prime[\rho;y]$,
and $\Fhat^\prime[\rho;\lambda x]  = \lambda \Fhat^\prime[\rho;x]$ for $\lambda \ge 0$.
Since $\langle v,\cdot\,\rangle$ is a linear functional, 
the sum $\Fhat^\prime[\rho;\cdot\,] + \langle v,\cdot\, \rangle$ is also sublinear.  
Since $\rho$ is the ground-state density of $v$, 
$\Fhat^\prime[\rho;\cdot\,] + \langle v,\cdot \rangle \ge 0$. 
We are trying to show this is an equality.  
Suppose not.  Then, there is some $\delta \rho \in {\mathcal K}$ such that
\[\Fhat^\prime[\rho;\delta \rho] + \langle v , \delta \rho \rangle > 0.\]
On the one-dimensional subspace spanned by $\delta \rho$, there is therefore a 
non-zero linear functional $\lambda$ such that $\lambda < \Fhat^\prime[\rho;\cdot\,] + 
\langle v,\cdot\, \rangle$.
A version of the Hahn-Banach theorem now implies the existence of an extension 
of $\lambda$, which we continue to denote by $\lambda$, 
such that 
\[\lambda(\cdot\,) \le \Fhat^\prime[\rho;\cdot\,] + \langle v,\cdot\, \rangle\] 
on all of ${\mathcal K}$.

The immediate object is to show that $\lambda$ acting on $\delta\rho\in{\mathcal K}$
is represented as
\[
\lambda(\delta \rho) = \int -w({\bm x}) \, \delta\rho({\bm x})\, d{\bm x}.
\]
This is not automatic.  We have only shown that $\lambda$ is a {\em linear}
fuctional, not that it is continuous.  But, for density perturbations nonzero
only on a bounded set $\Sigma$, consisting of a finite number of cells of $\Part$, 
there must be such a $w$, since then the functional is on a finite-dimensional
space.  Such $w$'s for different $\Sigma$'s must agree on overlaps and they
can be patched together to yield a single $\Part$-measurable function.
So, the $w$ representation holds on all of ${\mathcal K}$.

Thus, 
\[
0 \le \Fhat^\prime[\rho;\cdot\,] + \langle v + w, \cdot\, \rangle.
\]
This relation appears to say that $\rho$ is the ground-state density for 
$v+w$.  If $w$ is bounded below, that is correct and we can therefore apply
the coarse-grained Hohenberg-Kohn theorem to conclude that $w\equiv 0$.
But if $w$ is not bounded below, that argument is not immediately applicable.
In that case, consider, the density perturbation 
$\delta \rho_{\bm R}$ which just adds $1$ to some arbitrary cell ${\bm R}$.
Since $\Fhat$ is bounded above and below, and is convex,
$-w({\bm R}) - v({\bm R}) \le \Fhat^\prime[\rho;\delta \rho_{\bm R}] 
\le \Fhat[\rho + \, \delta \rho_{\bm R}] - \Fhat[\rho] \le \Fmax$.
Therefore, since $\bm R$ is arbitrary, $v+w$ is bounded below, so the Hohenberg-Kohn 
argument actually does apply to $v+w$, after all, showing that $w \equiv 0$.
\end{proof}

\noindent{\bf Proof of Theorem \ref{Diff-Thm}.}
Case (a) is clear.  So suppose there is $s > 0$ such that 
$\rho+s\delta \rho \in {\X}^+$.  
Renormalizing $\delta \rho$ if necessary, assume without loss that
$\rho+s\delta \rho \in {\X}^+$ for $s \le 1$.  

Fix $\epsilon > 0$.  Using Theorem \ref{Continuity-Thm}, find $\delta$ such that 
\beq
\left|\Fhat[\rho^\prime] - \Fhat[\rho]\right| < \epsilon/4, \quad \rho^\prime \in B_\delta(\rho),
\eeq
and then a bounded region $\Sigma$ large enough that 
$\| \delta \rho{\rest_{\Sigma^c}}\|_1 < \delta$.  

Now, we split $\Fhat[\rho + s\, \delta \rho] - \Fhat[\rho]$ as
\begin{eqnarray}
\Fhat[\rho + s\, \delta \rho] - \Fhat[\rho] &= &
\Big( \Fhat[\rho + s\, \delta \rho] - \Fhat[\rho + s\, \delta \rho\rest_\Sigma]\Big)
\nonumber \\
&+& 
\Big( \Fhat[\rho + s\, \delta \rho\rest_\Sigma] - \Fhat[\rho]\Big), 
\label{split} 
\end{eqnarray}
and deal with the two terms (``1st line'' and ``2nd line'') on the right-hand
side following a strategy illustrated by Figure \ref{nbhd-fig}.

\begin{figure}
\includegraphics[height=30mm]{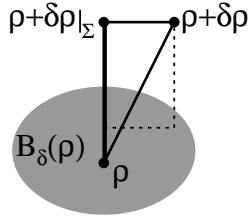}
\caption{Schematic of construction in proof of Theorem \ref{Diff-Thm}.}
\label{nbhd-fig}
\end{figure}

For the 2nd line, since $\delta \rho\rest_{\Sigma}$ has bounded support,
Lemma \ref{Baby-Diff} says that 
\[
{\Fhat[\rho+ s\, \delta \rho\rest_\Sigma]} - \Fhat[\rho]  
= -s \int v \delta \rho\rest_\Sigma \, d{\bm r} + {o}(s).
\]
As for the 1st line, convexity of $\Fhat$ implies that
\[
\Fhat[\rho+s\,\delta \rho] \le 
(1-s) \Fhat[\rho+s\,\delta \rho\rest_{\Sigma}] + 
s \Fhat[\rho+s\,\delta \rho\rest_{\Sigma} + \delta \rho\rest_{\Sigma^c}],
\]
so that the 1st line satisfies
\beq
\Fhat[\rho+s\,\delta \rho] - \Fhat[\rho+s\,\delta \rho\rest_\Sigma] 
 \le 
s \Big( \Fhat[\rho+s\,\delta \rho\rest_{\Sigma} + \delta \rho\rest_{\Sigma^c}]
- \Fhat[\rho+s\,\delta \rho\rest_{\Sigma}] \Big).
\label{1st-line-reprise}
\eeq
For $s$ small enough (Fig. \ref{nbhd-fig}),
both $\rho+s\,\delta \rho\rest_{\Sigma}$ and 
$\rho+s\,\delta \rho\rest_{\Sigma} + \delta \rho\rest_{\Sigma^c}$ are in
$B_\delta(\rho)$, so that their intrinsic energies do not differ by more than $\epsilon/2$.
Injecting that fact into inequality (\ref{1st-line-reprise}),
\[
\Fhat[\rho+s\,\delta \rho] - \Fhat[\rho+s\,\delta \rho\rest_{\Sigma}] \le s \epsilon/2.
\]
Putting everything back into Eq. (\ref{split}),
$
\Fhat[\rho + s\, \delta \rho] - \Fhat[\rho] \le 
s \left( {\epsilon}/{2} - \int v \delta \rho\rest_\Sigma \, d{\bm x}\right) + {o}(s),
$
showing that 
\[
\Fhat^\prime[\rho; \delta \rho] \le {\epsilon}/{2} - \int v\, \delta \rho\rest_\Sigma \, d{\bm x}.
\]
In the limit $\Sigma \nearrow \Part$, the integral tends to
$\int v\, \delta \rho \, d{\bm r}$ whether or not $\delta \rho$ is in the domain of
$v$.  If it is not, the value of the integral is unambiguously $-\infty$.
Taking now the limit $\epsilon \to 0$ gives 
$\Fhat^\prime[\rho; \delta \rho] \le - \int v \delta \rho\, d{\bm x}$.
The variational principle already secured the opposite inequality, thus
we have equality:
$\Fhat^\prime[\rho; \delta \rho] =- \int v\, \delta \rho\, d{\bm x}$.
\qedsymbol

The final lemma-theorem pair in this section is aimed at getting something
resembling continuity of $v[\rho]$ as a function of $\rho$.

\begin{lem}[`Local continuity' of potential] 
\label{local-niceness}
The potential $v[\rho]$ is a continuous function of $\rho$ with respect to the $L^1$ 
topology on ${\XPP}$ and the product topology on ${\V}$. 
\end{lem}

\begin{proof}
To say that $v[\rho]$ is continuous with respect to the product topology 
is to say that $v[\rho]\rest_\Sigma$ is a continuous function of $\rho$.
Thus, fix a bounded region $\Sigma$ and view 
\[
\Fhat[\rho] = \Fhat[{\rho\rest_\Sigma} + \rho\rest_{\Sigma^c}] 
\]
as a function of the two variables $\rho\rest_\Sigma$ and $\rho\rest_{\Sigma^c}$. 

For fixed $\rho\rest_{\Sigma^c}$, $\Fhat[\rho]$ is a bounded convex differentiable 
function (according to Lemma \ref{Baby-Diff}) 
of $\rho\rest_\Sigma$.
Temporarily denote the derivative with respect to $\rho\rest_\Sigma$ by
\[
D_\Sigma \Fhat[\rho]. 
\]
According to a theorem of convex analysis (Appendix \ref{cvx-appendix} (6)),
boundedness of $\Fhat$ and finite-dimensionality of the variable $\rho\rest_\Sigma$
imply that $D_\Sigma \Fhat[\rho]$ is a continuous function of ${\rho\rest}_\Sigma$.  

On the other hand, for fixed ${\rho\rest}_\Sigma$, $\Fhat[\rho]$ is a continuous
function of $\rho\rest_{\Sigma^c}$ according to Theorem \ref{Continuity-Thm}.
Again from boundedness and finite-dimensionality conclude that $D_\Sigma \Fhat[\rho]$ 
is {\em also\/} a continuous function of $\rho\rest_{\Sigma^c}$.
For, if not, there would be a sequence of values $\tau_n$ converging to
$\rho\rest_{\Sigma^c}$ such that $\Fhat[{\rho\rest_\Sigma} + \tau_n]$ 
converges to $\Fhat[{\rho\rest_\Sigma} + \rho\rest_{\Sigma^c}]$,
but $D_\Sigma \Fhat[{\rho\rest_\Sigma} + \tau_n]$ 
does not converge to $D_\Sigma \Fhat[{\rho\rest_\Sigma} + \rho\rest_{\Sigma^c}]$.
According to Appendix \ref{cvx-appendix} (5), that cannot happen.

But, the derivative of $\Fhat[\rho]$ with respect to $\rho\rest_\Sigma$ is precisely
$v[\rho]\rest_\Sigma$.
So the previous statement says that $v[\rho]\rest_\Sigma$ is continuous with respect to both 
$\rho\rest_{\Sigma}$ and $\rho\rest_{\Sigma^c}$, which is to say,
it is continuous with respect to $\rho$.
\end{proof}

\noindent{\bf Proof of Theorem \ref{V-Cont-Thm}.}
 Suppose not.  Then, there exists $\epsilon > 0$ and a sequence 
$\rho_m \to \rho$ such that 
\beq
\| v[\rho_m]\rho_m - v[\rho] \rho \|_1 > \epsilon,\quad \mbox{\rm for all }m.
\label{nonconvergence}
\eeq
We will derive a contradiction.

According to our convention [Eq. (\ref{offset-convention})] on the constant offset of $v[\rho]$,
\beq
\int_{\Sigma^c} v[\rho_m] \rho_m \, d{\bm x} = -\Fhat[\rho_m] - \int_\Sigma v[\rho_m] \rho_m\, d{\bm x},
\label{energy-split}
\eeq
for any bounded region $\Sigma$.
As $m\to\infty$, the first term on the right-hand side converges to $\Fhat[\rho]$ by Theorem 
\ref{Continuity-Thm}.
Also, by Lemma \ref{local-niceness}, the integral on the right-hand side converges to
$\int_\Sigma v[\rho]\, \rho\, d{\bm x}$, since the integration is over a finite region.
Actually, that Lemma shows more.
The contribution to Inequality (\ref{nonconvergence}) from integration over $\Sigma$ tends 
to zero as $m \to \infty$.  

So, as $m\to\infty$, both terms on the right-hand side of Eq. (\ref{energy-split}) 
tend to their counterparts with $\rho_m$ replaced by $\rho$.
Since the entire equation also holds with that substitution,
the left-hand side must converge to $\int_{\Sigma^c} v[\rho]\, \rho \, d{\bm x}$. 
Abbreviating $I_m^+(\Sigma) =  \int_{\Sigma^c} v_m^+ \rho_m \, d{\bm x}$,
$I_m^-(\Sigma) =  \int_{\Sigma^c} v_m^- \rho_m \, d{\bm x}$,
and similarly with no $m$ subscript, that means
\[
I_m^+(\Sigma) - I_m^-(\Sigma) \to I^+(\Sigma)-I^-(\Sigma).
\]
Now, both $I^+(\Sigma)$ and $I^-(\Sigma)$ tend to zero as $\Sigma \nearrow {\mathbb R}^3$.
This shows that 
\[
I_m^+(\Sigma) \to I_m^-(\Sigma),\;\; \mbox{\rm as } \Sigma\nearrow{\mathbb R}^3 
\,\mbox{\rm and }\, m\to\infty. 
\]

On the other hand, as has been remarked,
inequality (\ref{nonconvergence}) is almost entirely carried by
$\Sigma^c$ as $m\to\infty$, for any $\Sigma$. 
So, for large enough $\Sigma$ and $m$,
\[
I_m^+(\Sigma) + I_m^-(\Sigma) > \epsilon/2.  
\]
The last two displays together show that, for large enough $\Sigma$ and $m$,
{\em both\/} $I_m^+(\Sigma)$ and $I_m^-(\Sigma)$ are bounded away from
zero.  But that is clearly impossible. 
$v_m$ is bounded below ($v_m^- \le \Fmax$), so that
\[
I_m^-(\Sigma) \le \Fmax \int_{\Sigma^c} \rho_m\, d{\bm x}. 
\]
Yet, $\int_{\Sigma^c} \rho_m \, d{\bm x}$ tends to zero as $\Sigma \nearrow {\mathbb R}^3$ 
and $m\to\infty$, since $\rho_m {\to} \rho$.
This contradiction finishes the proof.
\qedsymbol

\section{Kohn-Sham Theory}
\label{KS}

This section concerns the implementation of Kohn-Sham theory within
the coarse-grained model.  Since all densities in $\XNPP$ are
both interacting and non-interacting V-representable, the endeavor
is off to a good start.  
Two additional issues are a little subtlety in the definition of
the Hartree energy, and the relation $v_{xc} = \delta E_{xc}/\delta\rho$.

Of course, the Kohn-Sham decomposition of the intrinsic energy for 
$\rho\in {\X}^+$ looks just like that for the fine-grained theory, 
Eq. (\ref{KS-decomp}):  $F[\rho] = T_s[\rho]+E_H[\rho]+E_{xc}[\rho]$.
Only the interpretation is changed.
$E_{xc}[\rho]$ is defined only on ${\XNP}$, since
outside, both $F[\rho]$ and $T_s[\rho]$ are $+\infty$.
Similarly, Eq. (\ref{KS-potls}) is taken over and used to define the
exchange-correlation potential, as
$v_{xc}[\rho] \defeq v_s[\rho] - v[\rho] -\phi[\rho]$ on ${\XNPP}$.

The definition of the Hartree energy requires supplementation, 
because a coarse-grained density only fixes the total particle number in each 
cell.  A choice must be made here which does not arise in the fine-grained 
theory, and there seem to be two possibilities.  Different choices here would 
imply slightly different $E_{xc}$.
One possibility is to use a density which is {\em uniform} throughout each cell.
In that case, we would begin using this surrogate for a coarse-grained equivalence 
class in a more explicit manner, and the Hartee energy would depend only on
the coarse-grained equivalence class.
The second possibility is to use one of the fine-grained densities 
which minimize the intrinsic energy.  This certainly seems natural in some
ways, but it raises problems.  First, which intrinsic energy, $F$ or $T_s$?
Even if that is decided, say for $F$, there may be multiple fine-grained
densities which minimize $F$, calling for another choice.
One of the nice features of the Hartree energy is its explicitness.
That would be lost with this choice, and possibly the ability to
establish continuity of $E_H$ as well.  For all these reasons, we
will take the first choice: the Hartree energy is calculated according
to a uniform distribution of charge in cells.

Off ${\X}^+$, $\rho$ can be negative, and if we interpret that as
positive electrical charge, $E_H$ becomes well-defined on all of ${\X}$.  
It is also\cite{Lieb-Loss} convex, bounded below by zero, and bounded above 
by a multiple of $\|\rho\|_1^2$, the multiple being determined by the 
energy to charge a single cell.  These properties imply that $E_H$ is
continuous as a function of $\rho \in {\X}$.

The derivative of $E_H$, which will be denoted $\phi$ is explicitly computable as 
an element of ${\X}^*$.
From $\varphi({\bm x}) = e^2 \int {\rho({\bm y})}/{|{\bm x}-{\bm y}|} \, d{\bm y}$,
$\phi$ is obtained by averaging over cells.
The map $\rho \mapsto \phi[\rho]$ is linear, and it is not difficult to see that 
it is continuous from ${\X}$ to ${\X}^*$, which is to say, $\|\phi\|_\infty \le C\|\rho\|_1$.
Thus, $\phi$ is a Fr\'echet derivative of $E_H$.
The Hartree potential behaves better than does $v[\rho]$.

Now we turn to the exchange-correlation potential.
It follows immediately from the definition that $\rho \mapsto v_{xc}[\rho] \rho$ 
is $L^1$-continuous, since the other three potentials have this property.
Is $v_{xc}$ the derivative of $E_{xc}$?  More precisely, does it
coincide with the directional derivative $E^\prime[\rho;\cdot\,]$?
Consider what can go wrong.  If $\delta \rho$ is such
that $\rho + s\, \delta\rho$ is not in ${\XNP}$ for any $s > 0$,
then $E_{xc}^\prime[\rho;\delta\rho]$ has no value, finite or infinite,
because $E_{xc}$ is not defined off ${\XNP}$.
Otherwise, we should have
$E_{xc}^\prime[\rho;\delta\rho] =
F^\prime[\rho;\delta\rho] - T_s^\prime[\rho;\delta\rho] - E_H^\prime[\rho;\delta\rho]
= \langle v[\rho],\delta\rho\rangle - \langle v_s[\rho],\delta\rho\rangle
-\langle \phi[\rho],\delta\rho\rangle$.
This is correct as long as at most one term on the right hand side is
infinite.  Either $F^\prime[\rho;\delta\rho]$ or $T_s[\rho;\delta\rho]$ might be $-\infty$.  
A condition which can be imposed directly on $\rho$ to make sure 
$E_{xc}^\prime[\rho;\delta\rho]$ and $\langle v_{xc},\delta\rho\rangle$ both
exist and are equal is that $\rho + s\, \delta\rho$ is in ${\XNP}$ for $s$ 
in some open interval around zero.
The fairly solid status of the coarse-grained exchange-correlation potential is
in stark contrast to its fine-grained counterpart.  

\section{Multiscale Model and Limits of Zero Coarse-Graining Scale}
\label{continuum}

The previous section aimed to show that coarse-graining cures some
of the bad behaviors of the continuum theory.
But not everything about the continuum theory is bad, and it is desirable that
coarse-grained models be good approximations to it in certain respects.  
This immediately leads us to ask about continuum limits.  What happens as 
the coarse-graining scale is taken to zero?  Is the continuum theory 
recoverd as smoothly as possible?  Are there misleading limits?  
These are the concerns of the present section.

\subsection{ideas}
\label{continuum-ideas}


A straightforward way to approach a continuum limit is to use a 
sequence of ever-finer acceptable partitions
${\Part}_n$, $n=1,2,3,\ldots$ of the sort introduced in Section \ref{CG},
with the maximum cell diameter of ${\Part}_{n}$, denoted $D_n$, 
tending to zero as $n\to\infty$.  
For the sake of the Poincar\'e inequality in Thm. \ref{F-limit}, we
impose the technical condition that the cells be convex.  
Finally, in order to have the collections of coarse-grained densities
strictly increasing with level, we require the cells of ${\Part}_{n+1}$ 
to be obtained by subpartitioning those of ${\Part}_{n}$.
For example, these requirements are satisfied if we take the cells of ${\Part}_n$ to be 
those of a simple triclinic lattice generated by lattice vectors $2^{-n}{\bm a}_k$ 
with ${\bm a}_1$, ${\bm a}_2$, and ${\bm a}_3$ noncolinear.

As in Section \ref{CG}, each $\Part_n$ gives rise to a space 
${\X}^n$ of equivalence classes of densities.  The object now is to 
make contact with the theory sketched in Section \ref{fine-grained}, 
where $X = L^1\cap L^3$ and its subset $\JN$ figured prominently, so 
we will equip ${\X}^n$ with the $L^1\cap L^3$ norm as well, instead of the
$L^1$ norm as we did with the single-scale model.  The point is 
that this does not require any essential change to what was done in the 
previous two sections because for any {\em fixed} ${\X}^n$, the $L^1\cap L^3$ 
and $L^1$ norms are equivalent: $\|f\|_3  \le c_n \|f\|_1$ for $f\in{\X}^n$.
For a sequence on coarse-graining scales tending to zero, 
though, convergence with respect to $L^1\cap L^3$ is more stringent 
than $L^1$ convergence. 
In studying the single-scale model, it was convenient and harmless to 
represent elements of ${\X}$ by $\Part$-measurable functions (that is, 
constant on cells).  In the multiscale setting, this conflation 
is not so innocuous.  Therefore, we introduce the injective isometry
\beq
\iota_n: {\X}^{n} {\hookrightarrow} {X},
\eeq
which takes $\rho \in {\X}^n$ to the unique ${\Part}_n$-measurable function
belonging to the equivalence class in $X$ associated to $\rho$.

For any $n$, each $\dotro\in X$ belongs to one equivalence class
corresponding to an element of ${\X}^n$, which will be denoted by $\pi_n \dotro$.  
As a notational cue, elements of $X$ will generally be denoted by a rho with
tucked-in tail ($\dotro$), and elements of one of the ${\X}^n$ by a normal rho ($\rho$).
We refer to $\pi_n\dotro$ as the (scale-$n$) {\it projection} of $\dotro$.
This nomenclature is doubly justified, since $\iota_n {\X}^n$ is a closed 
subspace of $X$ and $\iota_n \pi_n \dotro$ is precisely the $L^2$ orthogonal projection
of $\dotro$ onto this subspace.  We wish to think of coarse-grained densities
as formal objects independent of $X$, and thus stop short of simply identifying
them with equivalence classes in $X$.  Consistently with the notation just introduced,
the equivalence class corresponding to $\rho$ is $\pi_n^{-1} \rho$.

Each ${\X}^{n}$ has its own intrinsic energy functional $F^{n}$.
For $\rho \in {\X}^n$,
\[
F^{n}[\rho] \defeq \inf\{F[\dotro]: \dotro \in \pi_n^{-1} \rho \}.
\]
As a notational convenience, we write $F^n[\dotro]$ for $F^n[\pi_n\dotro]$.
As discussed in Section \ref{CG}, if $\rho \in ({\X}^n)_N^+$,
the infimum is attained at some fine-grained density, and possibly more than one, 
though such degeneracy is not generally expected.  
For $\rho \in ({\X}^n)_N^+$, the entire set of such minimizers will be denoted 
\beq
\Lambda_n \rho \defeq \{ \dotro \in \pi_n^{-1} \rho: F[\dotro] = F^n[\rho] \}.
\eeq
In particular, if $\rho$ is misnormalized or somewhere negative so that 
$F^n[\rho] = +\infty$, then $\Lambda_n \rho \defeq \emptyset$.
Thus, $\Lambda_n$ is a {\em set-valued\/} function or 
{\it multifunction}\cite{Aubin-Ekeland,Deimling,Aubin-Frankowska}.

The collection of all coarse-grained densities at all scales is denoted
\beq
{\X}^\infty \defeq \bigcup_{n=1}^\infty {\X}^n,
\eeq
and for $\rho\in {\X}^\infty$, the $\scale$ function is defined by 
$\rho \in {\X}^{\scale(\rho)}$.  The following chain of
inclusions, all of them strict, then holds:
\beq
\iota_1 {\X}^{1} {\subset} \iota_2{\X}^2\subset \cdots \subset \,
\iota {\X}^{\infty} = X\setminus\JN \, 
{\subset} {X}.
\eeq
On $\rho\in {\X}^\infty$, $\iota$ acts as $\iota_{\scale(\rho)}$.
The scale index on $\iota$ (and on $\Lambda$) is strictly unnecessary, 
but will be written at least when disambiguation is thereby provided. 
Note that $\iota_n {\X}^n$ is a closed vector subspace of $X$,
but $\iota {\X}^{\infty}$ is not a vector space at all.
It does not make sense to take a linear combination of coarse-grained
densities at different scales.
The fine-grained densities of interest are in $\JN$.  

%
Turning to potentials, we have a chain of spaces
\beq
{\V}^{1} {\subset} {\V}^2\subset \cdots \subset 
{\V}^{\infty} \defeq \bigcup_{n=1}^\infty {\V}^n,
\eeq
where ${\V}^n$ consists of ${\Part}_n$-measurable functions
bounded below.  
The coarse-grained theory of previous sections provides a 
representing potential map $v^n: ({\X}^n)_N^{++} \to {\V}^n$.
From the perspective of the fine-grained theory, 
$v^n[\rho]$ is the representing potential of $\Lambda_n\rho$,
and in making contact with that theory, the subspace
$({\X}^n)^* = \{ v \in {\V}^n: \|v\|_{X^*} < \infty\}$ 
of the potentials with finite $L^\infty \cap L^{3/2}$ norm, and understood
to be carrying that norm, is of more interest.  These nest as
\beq
({\X}^{1})^* {\subset} ({\X}^2)^*\subset \cdots \subset 
{\X^*}^{\infty} \defeq \cup_{n=1}^\infty ({\X}^n)^* \subset {X}^*.
\eeq
Note that for the potentials, there are no embedding maps 
analogous to $\iota$.  A coarse-grained potential {\em really is\/}
a ${\Part}_n$-measurable function.


Now we can begin to investigate how these ideas fit together. 
The first set of questions involve convergence, or lack thereof,
of $\iota \pi_n \dotro$ and $\Lambda \pi_n \dotro$ to $\dotro$.
The positive answers hold implications for convergence of $F^n[\pi_n \dotro]$.

\begin{thm}
\label{F-limit}
Given $\dotro$ in $\JN$, 
$\pi_n\dotro \stackrel{X}{\to} \dotro$.
$F^{n}[\dotro]$ is increasing with $n$, and 
$F^{n}[\dotro] = F[\Lambda_n\dotro] \nearrow F[\dotro]$ as $n\to\infty$.
Both $\Lambda \pi_n\dotro$ and $\iota \pi_n \dotro$ converge to $\dotro$
in $L^p$ norm as $n\to\infty$ for $1 \le p < 3$, with
$\| \dotro - \Lambda \pi_n\dotro \|_p, \| \dotro - \iota \pi_n\dotro \|_p \le
c D_n^{(3-p)/2} F[\dotro]^{(1+p)/4}$.
\end{thm}

Convergence $\| \iota \pi_n\dotro - \dotro\|_X \to 0$ 
is just a general property of $L^p$ functions under cell-averaging.
$L^2$ convergence is immediate from the description of $\pi_n$ as
an orthogonal projection.  Convergence in $L^p$ for other $p$ 
is proven by Jensen's inequality.
Although it is not exactly the sort of approximation coarse-graining was
designed to produce, it is not surprising that $\iota \pi_n\dotro$ tends to
$\dotro$ in $X$ as $n\to\infty$.  
It is much less obvious that $\iota \pi_n\dotro$ should converge to $\Lambda \pi_n \dotro$
since $\Lambda \pi_n\dotro$ is a moving target.   
We might just as well ask how spread out $\pi_n^{-1} \rho$ is, as a subset of $X$.  
Every element $\dotro \in \pi^{-1} \rho$ must match $\iota \rho$ down to scale
$\scale(\rho)$, but might differ arbitrarily much at ``subgrid'' scales.
However, there are limits imposed by $F[\dotro]$ through the 
lower bound in (\ref{Lieb-sandwich}).  Roughness at subgrid scales come 
at the cost of intrinsic energy, as intuition suggests.  
Using the bound with a Poincar\'e inequality enables us to prove a
bound on the $L^p$ norm $\|\iota \pi_n \dotro - \dotro\|_p$ 
for $1\le p < 3$ in terms of $F[\dotro]$:
\beq
\|\iota \pi_n \dotro - \dotro \|_p \le c_p D_n^2 (F[\dotro]/D_n^2)^{(1+p)/4}, \quad 1 \le p \le 3,
\label{pi_rho-convergence-rate}
\eeq
where $D_n$ is the minimum cell diameter in ${\Part}_n$.
Thus, the set of fine-grained densities in $\pi^{-1} \rho$ with low intrinsic energy,
$\pi^{-1} \rho \cap \{F < M\}$ is bounded in $L^p$ norm.  
Although the bound in (\ref{pi_rho-convergence-rate}) holds for
$p=3$, it does not show convergence in that case since the factors of $D_n$ disappear.
This clears the way to show that $\Lambda \pi_{n} \dotro \to \dotro$ in $L^p$.
For, it is a triviality that $F^1[\dotro] \le F^2[\dotro] \le F^3[\dotro] \le \cdots \le F[\dotro]$,
since each successive term represents a minimization with additional constraints 
(over a smaller set).  Since $F[\Lambda \pi_n\dotro] = F^n[\dotro]$, there is a common bound for
both $\|\dotro - \iota \pi_n\dotro\|_p$ and $\|\Lambda \pi_n\dotro - \iota \pi_n\dotro\|_p$ 
improving with $n$.   
We have noted that $F^n[\dotro]$ is increasing with $n$ {\em toward} 
$F[\dotro]$, but could conceivably fail to reach $F[\dotro]$ in the limit.  
But, since $F$ is $L^1$ lower semicontinuous and $\Lambda_n\dotro$ converges 
to $\dotro$ in $L^1$, the gap does close.  No intrisic energy somehow goes unaccounted 
for in the limit.  

We can broaden the scope of the above a little by considering
not only sequences going ``straight up'' the hierarchy of spaces along 
$\pi_n\dotro$, but also those converging toward that one.
A sequence $(\rho_j)_{j=1}^\infty$ in ${\X}^\infty$ is said to
converge to $\dotro\in\JN$ if $\|\dotro - \iota \rho_j\|_X \to 0$ as $j\to\infty$.
Since $\iota \pi_n\dotro$ is the $L^2$ projection of $\dotro$ onto $\iota{\X}^n$ and 
is at nonzero $L^2$ distance from it, $\dotro$ is at nonzero $X$-distance 
from $\iota{\X}^n$.  Thus, norm convergence of $\iota \rho_j$ to $\dotro$ implies that
$\scale(\rho_j) \to \infty$.  And therefore, 
since Thm. \ref{F-limit} has shown that $\iota \pi_n\dotro \to \dotro$,
the condition defining ``coarse-grained sequence converging to $\dotro$''
can equivalently be written as
$\| \iota (\rho_j - \pi_{n_j}\dotro) \|_X \to 0$, where
$n_j = \scale(\rho_j)$.

The convergence results of Thm. \ref{F-limit} can thus be generalized mildly to sequences 
$(\rho_j)_{j=1}^\infty \subset {\X}^\infty$ converging to $\dotro$, but the bounds
on $F^{n_j}[\rho_j]$  ($n_j \defeq \scale(\rho_j)$)
are not automatic in that case and must be imposed as hypotheses.  

\begin{cor}
\label{F-not-straight}
Let $(\rho_j)_{j=1}^\infty$ be a coarse-grained sequence converging to $\dotro$.
If $F_j \defeq F[\Lambda \rho_j] \le M < \infty$, then $\Lambda \rho_j \to \dotro$ 
in $L^p$, where $n_j = \scale(\rho_j)$.
If $\limsup F_j \le F[\dotro]$, then $F_j \to F[\dotro]$.  
\end{cor}

We now shift attention to convergence questions related to representing
potentials.
\begin{thm}
\label{V-limit}
$\dotro \in \JN$ is $X^*$-representable if and only if
there exists a coarse-grained sequence $(\rho_j)_{j=1}^\infty$ 
converging to $\dotro$
such that $\left\{ \|v^{n_j}[\rho_j]\|_{X^*} : 1 \le j < \infty \right\}$ 
is bounded ($n_j \defeq \scale(\rho_j)$).
In that case, $v^{n_j}[\rho_j] \to v[\dotro]$ in weak-$*$ sense and 
$F^{n_j}[\rho_j] \to F[\dotro]$.
Further, there exists such a sequence such that the convergence
${v}^{n_j}[\rho_j] \to v[\dotro]$ is in $X^*$-norm.
\end{thm}

Unfortunately, the tools do not seem to be at hand to say anything significant 
beyond the context of $X^*$-representability.  However, in the latter context
the closedness of the graph of $\partial F$ and the Ekeland variational 
principle, both discussed in Section \ref{fine-grained}, are powerful tools.
Suppose $(\rho_j)_{j=1}^\infty$ is a coarse-grained sequence converging to $\dotro \in \JN$,
with representing potentials $v^j = v[\Lambda \rho_j]$.
If $(v^j)_{j=1}^\infty$ is merely bounded in $X^*$-norm, then $\dotro$ is $X^*$-representable
and $v^j \to v[\dotro]$ in the weak-* topology.  At first, this might look like
a straight transcription of the closedness of $\graph \partial F$ discussed
in Section \ref{fine-grained}.  If it were the case that 
$\Lambda_{n_j} \rho_j$ converged to $\dotro$ in $X$, that would be correct.  
For, bounded subsets of $X^*$ are weak-* compact, and any weak-* cluster
point $v$ would be in $\partial F[\dotro]$.  Since the Hohenberg-Kohn
theorem guarantees that $\partial F[\dotro]$ is a singleton modulo constants,
the entire sequence would have to converge.  However, it does not
follow from $\iota \rho_j\to \dotro$ that $\Lambda_{n_j} \rho_j \to \dotro$.
This circumstance requires using the Ekeland variational principle.
Showing that $\dotro$ is almost a ground state of $v^j$ means that the
sequence $(\Lambda_{n_j} \rho_j,v^j)_1^\infty$ can be replaced by a new one
$(\tilde{\rho}_j,\tilde{v}^j)_1^\infty$ such that $\iota \tilde{\rho}_j \to \dotro$ and 
$\|\tilde{v}^j - v^j\|_{X^*}\to 0$.  Closedness of $\graph \partial F$
then applies in the ordinary way.

In the other direction, suppose that $\dotro\in\JN$ is $X^*$-representable.
Then, according to the theorem, this is accurately reflected in the 
coarse-grained hierarchy: there is a coarse-grained sequence 
$(\rho_j)_{j=1}^\infty$ converging to $\dotro$, such that $v^j \to v[\dotro]$
(in this case we get norm convergence of the potentials).  
The proof of this direction is perhaps a little more interesting.
To prove this, we find an approximation to $v[\dotro]$ in $({\X}^n)^*$
for which $\dotro$ is nearly a ground state.  Then we apply the 
Ekeland variational principle, but in ${\X}^n$, {\em not} in $X$.
This produces a sequence of pairs $(\rho_j,v^j)_1^\infty$ with the desired property.
The disappointing aspect of this half of the theorem is that we
cannot show that $v^j[\pi_{n_j}\dotro] \to v[\dotro]$.  But, there was
no good reason to suppose that true, anyway.

We finish this section with some reflection related to the computational 
and physical appropriateness and significance of topologies on ${\X}^\infty$.
In preparation, we have a grab-bag theorem about the multifunction 
$\Lambda_n$.
A remarkable aspect of the bound (\ref{pi_rho-convergence-rate})
is that it can be used to show that $\Lambda_n$ preserves compactness.
\begin{thm}
\label{Lambda-nice}
$\Lambda_n$ takes values in convex, $L^1$-compact sets.
If $K$ is $L^1$-compact in ${\X}^n$, then $\Lambda_n K$ is $L^1$-compact.
$\Lambda_n$ has a  ($L^1\times L^1$) closed graph and is
upper semicontinuous (meaning, given $\rho$ and $\epsilon > 0$, there is
there is $\delta > 0$ such that $\|\rho^\prime - \rho\|_1 < \delta$
implies that all of $\Lambda_n \rho^\prime$ is within $\epsilon$ of 
$\Lambda_n \rho$.)
\end{thm}
\begin{proof}
That $\Lambda_n\rho$ is convex is trivial.  
That it is $L^1$ compact is a special case of the compactness of 
$\Lambda_n K$, shown below.  
Closedness of the graph of $\Lambda_n$ is a simple consequence of
the $L^1$ lower semicontinuity of $F$.
For, suppose $(\rho_j, \dotro_j)_{j=1}^\infty$
is an $L^1\times L^1$ Cauchy sequence in $\graph \Lambda_n$.
Then $\rho_j \to \rho \in ({\X}^n)_N^+$ and $\dotro_j \to \dotro \in {X}_N^+$.
Also, $F[\dotro] \le \lim_{j\to\infty} F[\dotro_j] = F^n[\rho]$, where the
inequality follows from lower semicontinuity of $F$ and the equality from
continuity of $F^n$.  Since $\dotro \in \pi_n^{-1} \rho$, 
$F[\dotro] \ge F^n[\rho]$, so that actually $F[\dotro] = F^n[\rho]$.  
In other words, $\dotro \in \Lambda_n \rho$, as was to be shown.

To see that $\Lambda_n$ preserves compactness, let $K \subset ({\X}^n)_N^+$ be $L^1$ compact.
Since $F^n$ is continuous, it is bounded above on $K$, say by $c$.   
Lemma \ref{Lambda-nice} in the next subsection and the previous paragraph 
then show $\Lambda_n K$ to be a closed totally bounded set, hence compact.

Suppose upper semicontinuity failed at $\rho$.  Then, there would be
a sequence $\rho_j \to \rho$ and $\dotro_j \in \Lambda_n \rho$ such
that every $\dotro_j$ was at $L^1$ distance greater than
$\epsilon$ from $\Lambda_n\rho$.  But, $\{\rho_j : 1 \le j < \infty \} \cup \{\rho\}$
is compact.  So $\{\dotro_j: 1 \le j < \infty\}$ is relatively compact, and
therefore by closedness of $\Lambda_n$ contains a subsequence converging to
an element of $\Lambda_n\rho$.  The contradiction proves upper semicontinuity.
\end{proof}
This theorem shows that, for compact sets $K\subset{\X}^n$, not only
is $\Lambda_n K$ not very far from $K$ (Theorem \ref{F-limit}), it is also
not much larger.

There are a couple of ways to look at the 
a single space ${\X}^n$, the $L^1$ or $L^1\cap L^3$ 
metrics are pullbacks via $\iota$ from $X$.

One way to look at the $L^1$ or $L^1\cap L^3$ metrics on a single space 
${\X}^n$ is to recognize that they correspond to the same topology, namely 
the weakest topology which makes the particle number in each cell of $\Part_n$ 
and the total particle number continuous.
This topology seems natural when considering coarse-grained densities.
Alternatively, we recognize them as pullbacks via $\iota$ or $\pi_n^{-1}$
from $X$.  While the connection with $\iota$ is fairly obvious, it is
also the case that
\[
d_1^{(n)}(\rho,\rho^\prime) \defeq
\| \rho,\rho^\prime \|_1 = 
\inf\{\|\dotro - \dotro^\prime \|_1 : \dotro \in \pi_n^{-1} \rho,
\, \dotro^\prime \in \pi_n^{-1} \rho^\prime\},
\]
and similarly for the $X$ ($L^1\cap L^3$) norm.
That is, the distance is just the distance between the sets (equivalence classes)
$\pi_n^{-1} \rho$ and $\pi_n^{-1} \rho^\prime$.

On the other hand, we have some tendency to see $\rho$ also as a kind
of surrogate for the low-intrinsic-energy densities in $\pi_n^{-1}\rho$,
specifically $\Lambda_n\rho$.  After all, that is what the definition
of $F^n$ was all about.  This inclines us to consider the metric
\[
d_\Lambda^{(n)}(\rho,\rho^\prime) = \inf\{\|\dotro - \dotro^\prime \|_1 : \dotro \in \Lambda_n\rho,
\, \dotro^\prime \in \Lambda_n\rho^\prime\}
\]
on $({\X}^n)_N^+$.  (Only $({\X}^n)_N^+$ can be metrized this way since $\Lambda$ 
is null-valued outside.)  It is clear that this is at least as strong as the
$d_1$ metric.   
What is interesting is that it is {\em topologically equivalent} to $d_1^{(n)}$.
In other words, $d_\Lambda^{(n)}(\rho,\rho^\prime)$ is continuous with respect to
the $d_1$ topology.  This conclusion follows from Thm. \ref{Lambda-nice}.


Moving from the single-scale of ${\X}^n$ to the multi-scale setting of ${\X}^\infty$,
uncovers a new oddity. 
Earlier in this section, for purposes of establishing the coarse-grained
model as a good approximation of the fine-grained theory, distances between 
coarse-grained densities at different scales were computed as the $X$-norm distance between
their images under $\iota$ or $\pi^{-1}$.
Arguably, this is somewhat too crude.  Certainly, it results in
some coarse-grained densities at different scales being at zero 
distance from one another, which does not seem a desirable outcome.
Likely, the analogous extension of $d_\Lambda^{(n)}$ does not have
that problem, but the uncertainty just points to a different difficulty,
which is that of computing $d_\Lambda$.

From a computational perspective, a coarse-grained density is a description 
with a degree of precision related to the scale and it seems that 
the $\scale$ function should be continuous.
An example of a metric which achieves that is
\beq
d(\rho,\rho^\prime) = \|\iota \rho - \iota\rho^\prime\|_X + 
|\scale(\rho)-\scale(\rho^\prime)|.
\eeq
It should be noted that all the results of Section \ref{CG}
are true for ${\X}^\infty$ under this metric.
There is no mathematical depth whatever to this metrization of ${\X}^\infty$,
and it is effectively the same as just treating the different scales as
incomparable.  But maybe the mere construction of ${\X}^\infty$ will cure 
us of the idea that some fundamental length scale is inherent in coarse-graining {\it per se}.

\subsection{deferred proofs}
\label{continuum-proofs}


\noindent{\bf  Proof of Theorem \ref{F-limit}.}
The convergence of $\iota \pi_n\dotro$ to $\dotro$ in $X$ norm 
hinges on Jensen's inequality. 
Since $y \mapsto |y|^p$ is convex for $1 \le p < \infty$,
\[
\left( \frac{1}{|\Omega|} \int_\Omega \dotro \, d{\bm x} \right)^p
\le 
\frac{1}{|\Omega|}\int_\Omega \dotro^p \, d{\bm x},
\]
for each cell $\Omega$ of a partition (volume $|\Omega|$).
Summing the left-hand side over cells yields $\|\pi\dotro\|_p^p$,
and summing the right-hand side yields $\|\dotro\|_p^p$.
Thus, the operator $\dotro \mapsto \dotro - \iota \pi_n\dotro$ is $L^p$ bounded,
with bound 2.  Since $C_c({\Bbb R}^3)$ (compactly supported continuous 
functions) is dense in $L^p$, we can find $f\in C_c({\Bbb R}^3)$ with
$\|f - \dotro\|_p < \epsilon/2$, so that 
$\|(f-\iota \pi_n f) - (\dotro - \iota \pi_n \dotro)\|_p < \epsilon$ for all $n$.  But, 
$f$ is {\em uniformly} continuous, being compactly supported, so that
$\|f - \iota \pi_n f\|_p \to 0$ as $n\to\infty$.  Thus, $\|\iota \pi_n\dotro - \dotro\|_X \to 0$.

A quantitative estimate of $\|\Lambda \pi_n\dotro - \dotro\|_p$ and $\|\iota \pi_n\dotro - \dotro\|_p$
hinges on an $L^1$ Poincar\'e\cite{Adams} (or Poincar\'e-Wirtinger\cite{Attouch}) 
inequality requiring convexity\cite{Acosta-Duran03,Bebendorf03}, but not
regularity, of the domain.
If $f$ is a function such that it and its (distributional) gradient $\nabla f$ are integrable 
over the convex bounded region $\Omega$ with diameter $\mathrm{Diam}(\Omega)$
(i.e., $f\in W^{1,1}(\Omega)$),
\[ 
\int_\Omega |f - \langle{f}\rangle_\Omega| \, d{\bm x} \le \frac{\pi}{2} \, \mathrm{Diam }(\Omega)
\int_\Omega |\nabla f| \, d{\bm x},
\] 
where $\langle{f}\rangle_\Omega$ is the mean of $f$ over $\Omega$.
Applying this to $\dotro$ and $\Lambda \pi_n\dotro$ on each cell and summing the
results, 
\beq 
\int |\dotro - \iota \pi_n \dotro| \, d{\bm x} \le \frac{\pi}{2} \, {D_n} \int |\nabla \dotro| \, d{\bm x},
\eeq 
and similarly for $\Lambda \pi_n\dotro$, more precisely (as should be understood for such
a locution) for each element of $\Lambda \pi_n\dotro$.
To make use of this, we bound $\|\nabla\dotro\|_1$ by using inequalities 
(\ref{Lieb-sandwich},\ref{L3-Sobolev}) and the Cauchy-Schwarz inequality, as
\begin{eqnarray}
\int |\nabla \dotro| \, d{\bm x} &= &
\int 2\dotro^{1/2} |\nabla \dotro^{1/2}| \, d{\bm x}
\nonumber \\
& \le &
2 \left( \int \dotro\, d{\bm x}\right)^{1/2}
\left( \int |\nabla \dotro^{1/2}|^2 \, d{\bm x} \right)^{1/2}
\nonumber \\ & \le & c^\prime N F[\dotro]^{1/2}.
\end{eqnarray}
Therefore, since all the $F^n[\dotro] = F[\Lambda \pi_n\dotro]$ are bounded above by $F[\dotro]$,
both $\|\dotro - \iota \pi_n \dotro\|_1$ and $\|\Lambda \pi_n\dotro - \iota \pi_n\dotro\|_1$
are bounded by $c D_n F[\dotro]^{1/2}$.
We obtain not only convergence, but a bound on the rate of convergence.

To extend this to $L^p$ for $1<p<3$, we make use of a H\"older inequality. 
For $f,g \in L^1\cap L^3$,
\begin{eqnarray}
\|f - g \|_p 
& \le &
\|f - g\|_3^{{(p-1)}/{2}} \|f - g\|_1^{{(3-p)}/{2}} 
\nonumber \\
& \le &
2^{{(p-1)}/{3}} \left(\|f\|_3^3 + \|g\|_3^3\right)^{{(p-1)}/{6}} \|f - g\|_1^{{(3-p)}/{2}}. 
\nonumber
\end{eqnarray}
Since each $\|\Lambda \pi_n \dotro\|_3$ as well as $\|\dotro\|_3$ is bounded by a 
constant times $F[\dotro]$, combining this with $L^1$ convergence of
$\Lambda \pi_n \dotro$ {to} $\dotro$ implies ${L^p}$ convergence as well.

Finally, we prove upward convergence of $F^n[\dotro]$ to $F[\dotro]$.
Since $F^{n}[\dotro]$ represents a minimization with increasing constraints as $n$ increases, 
$F^{1}[\dotro] \le F^{2}[\dotro]\le F^3[\dotro]\le \cdots \le F[\dotro]$ is trivial.
On the other hand, as discussed in Section \ref{fine-grained}, $F$ is $L^1$ lower 
semicontinuous.  Since $F^n[\dotro] = F[\Lambda \pi_n\dotro]$ and we have shown
that $\Lambda \pi_n\dotro$ converges to $\dotro$ in $L^1$, it follows that
$\liminf_n F^n[\dotro] \ge F[\dotro]$.  Thus, $F^n[\dotro] \nearrow F[\dotro]$.  

\qedsymbol

\noindent{\bf Proof of Theorem \ref{V-limit}.}

{\it ``If'' direction:}
To simplify notation, write $n_j$ for $\scale(\rho_j)$ and $v_j$ for $v^{n_j}[\rho_j]$.
Also, the $v^{n_j}[\rho_j]$ given by hypothesis satisfy our convention 
for fixing the constant offset in representing potentials
[Eq. (\ref{offset-convention})], but 
we otherwise lift that convention.  It will be restored at the end.

The idea is to show that $\dotro$ nearly attains the ground state energy of
$v_j$ for large $j$, so that the Ekeland variational principle can be called in.
First, we show that for $j$ large enough,
the potential energy of $\dotro$ in $v_j$ is not much more than that of
$\Lambda_{n_j} \rho_j$.  Since $v_j$ is ${\Part}_{n_j}$-measurable, 
it cannot distinguish $\dotro$ from $\pi_{n_j}\dotro$ or $\rho_j$ from $\Lambda \rho_j$.
Thus,
\[
\langle v_j, \dotro \rangle
= \langle v_j, \pi_{n_j} \dotro \rangle
= \langle v_j, \Lambda \rho_j \rangle + \langle v_j, (\pi_{n_j} \dotro - \rho_j) \rangle,
\]
so that, by hypothesis of boundedness of the potentials, 
for some $M$,
\beq
|\langle v_j, \dotro \rangle - \langle v_j, \Lambda_{n_j} \rho_j \rangle|
\le M \| \pi_{n_j} \dotro - \rho_j \|_X \to 0, \, \mathrm{as}\,\, j\to\infty.
\label{look-alike-to-vj}
\eeq
Second, we examine the intrinsic energy.  
Since $\Lambda_{n_j} \rho_j$ are the ground state densities of $v_j$,
and since as just shown, 
$\Lambda_{n_j} \rho_j$ and $\dotro$ look increasingly alike to $v_j$
as $j\to \infty$,
$F^{n_j}[\rho_j]$  cannot be much larger than $F[\dotro]$ for large $j$,
or eventually $\dotro$ would have lower total energy in $v^j$ than $\Lambda \rho_j$.
That is, 
\[
\limsup_{j\to\infty} F[\Lambda_{n_j} \rho_j] \le F[\dotro].
\]
Thus, by Corollary \ref{F-not-straight},
\beq
\lim_{j\to\infty} F[\Lambda_{n_j} \rho_j] = F[\dotro].  
\label{F-rhoj-to-Frho}
\eeq

It follows from (\ref{look-alike-to-vj}) and (\ref{F-rhoj-to-Frho})
that, given $\epsilon$, $F[\dotro] + \langle v_j, \dotro\rangle \le E[v_j] + \epsilon^2$
for $j \ge j(\epsilon)$.
Appealing to Ekeland's variational principle, conclude that there exists 
$(\tilde{\dotro}_\epsilon, \tilde{v}_\epsilon) \in \graph \partial F$ 
satisfying 
\[
\| \tilde{v}_\epsilon - v_{j(\epsilon)} \|_{X^*} \le \epsilon \quad \mbox{\rm and} \quad
\| \tilde{\dotro}_\epsilon - \dotro \|_X \le \epsilon.
\]

Suppose for the moment that $v_j$ converges to $v$ weak-$*$ in $X^*$. 
In that case, for $\epsilon \to 0$ $\tilde{v}_\epsilon \to v$ weak-$*$,
while $\tilde{\dotro}_{\epsilon} \to \dotro$ in $X$-norm.
From the closedness of $\graph \partial F$, it would follow that 
$\dotro$ is a ground state density of $v$. 

But, because $\{v_j\}$ is bounded in $X^*$, it is weak-$*$ compact,
hence has weak-$*$ accumulation points.  For any such, 
say $v$, we can extract a subsequence of $\{v_j\}$ converging 
to it weak-$*$ and the argument above applies.  On the other hand, the 
Hohenberg-Kohn theorem guarantees uniqueness of the representing potential for $\dotro$.
Putting it together, we conclude that the original sequence $v_j = v^{n_j}[\rho_j]$ 
converges weak-$*$ to $v$ and $\dotro$ is a ground state density of it.

Now check that $v$ satisfies the constant-offset convention, Eq. (\ref{offset-convention}).  
From weak-$*$ convergence, $\langle v_j, \dotro \rangle \to \langle v, \dotro \rangle$.
Together with the second displayed equation, this shows that
$\langle v_j, \rho_j \rangle \to \langle v, \dotro \rangle$.
On the other hand, $F^{n_j}[\rho_j] \to F[\dotro]$, so
$\langle v,\dotro\rangle + F[\dotro] = 0$.
\medskip

{\it ``Only if'' direction:} 
Let the tolerance $\epsilon > 0$ be given.
Compactly supported continuous functions are dense in
$X^*$, and any such can be uniformly approximated by some
element of $\cup_{n=0}^\infty ({\X}^n)^*$, so the latter is also 
dense in $X^*$.  Thus, there exists
$n$ and $v^\prime \in ({\X}^{n})^*$ with $\| v^\prime - v \|_{X^*}$ as small as desired.
$E[v]$ is Lipschitz continuous\cite{Lieb83} 
with respect to $X^*$-norm, so $n$ and $v^\prime$ can be chosen such that
$F[\dotro] + \langle v^\prime,\dotro\rangle = E[v] + \langle v^\prime - v,\dotro\rangle$ 
is a close to $E[v^\prime]$ as desired.
In fact, since  $v^\prime \in ({\X}^{n})^*$ and $F^n[\dotro] \to F[\dotro]$ as $n\to\infty$
by Theorem \ref{F-limit}, 
$n = n(\epsilon)$ and $v^\prime \in ({\X}^{n})^*$ can be chosen so that
$\| v^\prime - v \|_{X^*} < \epsilon/2$ and
\beq
F^n[\pi_n\dotro] + \langle v^\prime,\pi_n \dotro\rangle - E[v^\prime] < \epsilon^2/4.
\eeq

As in the ``if'' direction, appeal to Ekeland's variational principle again, but this
time using the pair $[{\X}^n,({\X}^n)^*]$ together with $F^n$, to conclude to
that there exists a pair 
$(\tilde{\rho},\tilde{v}) \in \graph \partial F^n \subset {\X}^n \times ({\X}^n)^*$ 
such that 
\[
\|\tilde{v} - v\|_{X^*} \le \epsilon
\]
and
\[
\|\tilde{\rho} - \pi_n \dotro \|_{X} \le \epsilon.
\]
Taking a sequence of values of $\epsilon$ tending to zero then produces a sequence 
as described in the statement of the Theorem.
\qedsymbol

\begin{lem}
Let $K \subset {\X}^n$ be toally bounded with respect to $L^1\cap L^p$ for $1\le p < 3$.  
Then for any $0 < c < \infty$, $\pi_n^{-1} K \cap \{F \le c\}$ is also totally bounded with 
respect to $L^1\cap L^p$.
\end{lem}
\begin{proof}
Since nothing in $\pi^{-1}\rho$ has finite intrinsic energy unless
$\rho$ is positive and properly normalized, we tacitly assume that all 
coarse-grained densities involved have those properties.
With $\epsilon > 0$ given, we show the existence of a finite $(3\epsilon)$-net
for $\pi_n^{-1} K \cap \{F \le c\}$.
According to Theorem \ref{F-limit}, there is $m \ge n$ such that
$\|\iota_m \rho - \dotro\|_{L^1\cap L^p} < \epsilon$ for any 
$\dotro \in \pi_m^{-1} \rho$ with $F[\dotro] \le c$.
Then, every $\rho$ in $\pi_n^{-1} K \cap \{F \le c\}$ is within $\epsilon$
of ($\iota$ of) some point of $(\pi^m_n)^{-1} K \cap ({\X}^m)^+$, the subset of $({\X}^m)^+$
which projects to $K \subset  {\X}^n$.
Thus, all that is needed is a finite $(2\epsilon)$-net 
for $(\pi^m_n)^{-1} K \cap ({\X}^m)^+$.  

Now note the following two simple facts.
For any scale $m$, the $L^p$ norm of an element of $({\X}^m)^+$ is bounded
in terms of its $L^1$ norm.  Also, a non-negative element of
$(\pi_n^m)^{-1}\rho$ has the same $L^1$ norm over any cell of ${\Part}_n$
as does $\rho$.  As a result, since $K$ is totally bounded by assumption,
there is a bounded ${\Part}_n$-measurable region $\Omega$ such that
the $L^1\cap L^p$ norm of every element of $(\pi_n^m)^{-1}K \cap ({\X}^m)^+$
over $\Omega^c$ (the complement of $\Omega$) is less than $\epsilon$.

Thus, it will suffice to restrict everything in $(\pi_n^m)^{-1}K \cap ({\X}^m)^+$
to $\Omega$ and find a finite $\epsilon$-net for that set.  But that is simple
because there are a finite number of degrees of freedom in those truncated
${\X}^m$ densities and each one is bounded due to total boundedness of $K$.
\end{proof}

\section{Conclusions}
\label{conclusion}

The results of this paper show that coarse-grained models of DFT
are mathematically well-behaved and are good approximations to the
standard, fine-grained, interpretation in that non-pathological
aspects are faithfully reflected.  These properties make the
coarse-grained models good regularizations of the fine-grained theory.
In such a r\^ole, coarse-grained models provide a controlled
arena in which to work, and also shed some light on the pathologies of
the fine-grained theory.

Beyond that, however, I have argued that the coarse-grained models
provide a superior {\it interpretation} of DFT.
The standard interpretation of DFT takes the somewhat reflexive view that 
a density involves an specification of infinitely fine spatial resolution.  
The coarse-grained interpretation understands each density to be specified 
with limited resolution, though there is no limit to that resolution.  
In DFT, one keeps track of spin density, but not degrees of 
freedom of the state which do not affect it.  Those are subject to an 
automatic energetic selection.  The step to coarse-grained DFT is very similar,
except that the degrees of freedom which the formalism gives over to
automatic energetic selection are the distribution of density within 
cells --- the short length scale degrees of freedom which resulted
in the fine-grained interpretation not being a good model.  
Paired in a natural way with a reinterpretation of external potential,
this yields a model of DFT where everything works {\em the way it is 
supposed to.}

\begin{acknowledgments}
I thank Cristiano Nisoli and Vin Crespi for suggestions on an early version of
the manuscript, John Clemens for a conversation about topological genericity,
and the Center for Nanoscale Science at the Pennsylvania State University
for financial support.
\end{acknowledgments}

\appendix
\section{functional analysis survival kit}
\label{fa-appendix}

This Appendix contains a quick review of the basic functional analysis used
in this paper.  Readers who need a quick reminder of a definition or notation
should find what they need.  With luck, readers unfamiliar with functional
analysis may find enough to follow the rough outlines of the ideas and arguments. 
Everything here is standard.

{\bf Topology and metric spaces.} Questions of continuity and convergence are
a major preoccupation in this paper.  The structure on a set $S$ which
makes these notions meaningful is {\it topology.}
{\it Neighborhood bases} are a convenient way to specify a topology.
A neighborhood ${\mathcal N}_x$ base for $x\in S$ is a collection $\{U_\alpha\}$ of
subsets of $S$ containing $x$, and satisfying the condition that,
whenever $U_\alpha,U_\beta \in {\mathcal N}_x$, then there is
some $U_\gamma \in {\mathcal N}_x$ such that 
$U_\gamma \subset U_\alpha,U_\beta \cap {\mathcal N}_x$.
Any superset of a member of ${\mathcal N}_x$ is a {\it neighborhood} of $x$
A subset $T$ of $S$ is {\it open} if, whenever $T$ contains a point $y$, 
it contains an entire neighborhood of $y$.
A sequence $(x_n)_{n=1}^\infty$ in topological space $S$ converges to
$x$ if, for any $U \in {\mathcal N}_x$, the entire tail of the sequence
from some $n$ (depending on $U$) onward is contained in $U$.
A subset $F\in S$ is said to be {\it closed} precisely when its
complement $S \setminus F$ is open.  The {\it interior} of a set $T$ is the
largest open set contained within $T$ and the {\it closure} of $T$ is
the smallest closed set containing $T$.

When one refers to ``the topology ${\mathcal T}$'' it means the collection of
open sets.  One topology ${\mathcal T^\prime}$ is called {\em stronger} than 
another, ${\mathcal T}$, if ${\mathcal T}\subset {\mathcal T^\prime}$.
For norm topologies (see Banach spaces below), the condition can be expressed by
$\|\cdot\| \le c \|\cdot\|^\prime$ for some $c>0$.  

Metric spaces are a particularly pleasant sort of topological space.
One choice of a neighborhood base of $x$ in a metric space $(S,d)$ is 
the collection of open balls $B(x,r) = \{y\in S: d(x,y) < r\}$ for
rational values of $r$.  Closed balls, or all real values of $r$ would
serve equally; they all specify the same topology.  
The metric space $S$ is even nicer if it is {\it complete}, meaning
that if the sequence $(x_n)_{n=1}^\infty$ is a Cauchy sequence, then
it actually converges (to some $x\in S$).

$T\subset S$ is {\it compact} if it has the property that for any
covering by open sets, $T \subset \cup_\alpha U_\alpha$, there
is a finite subfamily $U_{\alpha_1},\ldots U_{\alpha_n}$ which 
still covers $T$.  Roughly speaking, a compact set is 
almost topologically finite.  If $S$ is a metric space, then $T$
is compact if and only if it is complete and {\it totally bounded}.  
The latter condition means that,
given $\epsilon > 0$, $T$ can be covered by some finite collection of
balls of radius $\epsilon$.  For the metric space $T$, compactness
is equivalent to {\it sequential compactness} which means that
every sequence has a convergent subsequence (to a point in $T$).

{\bf Banach spaces.}  Let $V$ be a normed vector space, with norm 
$\|\cdot\|$.   The norm provides a metric via $d(x,y) = \|x-y\|$.
If $V$ is a complete metric space under this norm, then it is 
called a {\it Banach space}.
The classical Lebesgue spaces are Banach spaces defined by means
of integral norms.
The Lebesgue {$L^p$-norm} $\|\cdot\|_{p}$ for $1 \le p < \infty$ 
is given by
$
\|f\|_{p} 
\equiv \|f\|_{L^p} 
\defeq \left( \int |f({\bm x})|^p \, d{\bm x} \right)^{1/p}.
$
and the Lebesgue space $L^p({\Bbb R}^n)$ is the vector space of
all measurable functions on ${\Bbb R}^n$ with finite $L^p$-norm.
(In our case, we consider real functions.)  These spaces are Banach spaces.
The closed unit ball of a Banach space is compact in the norm topology 
if and only if the space is finite dimensional.
A Banach space (or even a topological space) is {\it separable}
if it contains a countable dense set.  All the classical $L^p$
spaces, except for $L^\infty$, are separable.

To each Banach space $X$ corresponds a {dual space}, denoted $X^*$,
comprising continuous linear functionals.  A common notation for the value 
of the functional $\lambda \in X^*$ on $x \in X$ is $\langle \lambda, x \rangle$.
An inequality of the form $|\langle u, f \rangle| \le \|u\|_{X^*} \|f\|_{X}$
expresses the content of {\it continuity} for a {\em linear} functional, 
and simultaneously serves to define a norm on $X^*$, which turns {\em it} into a Banach space.
The inequality also shows why continuous linear functionals are called {\it bounded}.

The duals of the Lebesgue spaces can be identified with Lebesgue spaces themselves.
For $1<p<\infty$, $L^p({\Bbb R}^n)^* = L^q({\Bbb R}^3)$ with ${1}/{q} + {1}/{p} = 1$
and the duality pairing given by an integral:
$\langle \lambda, u \rangle = \int \lambda(x) u(x) \, dx$.
For $p=1$, the formula gives $q = \infty$.  $L^\infty = (L^1)^*$ consists of functions which 
are essentially bounded, meaning that for $u \in L^\infty$, there is some
a number $\|u\|_\infty$ $M$ such that $|u({\bm x})| \le M$ off a set of measure zero.
$\|u\|_\infty$, the norm of $u$, is the smallest such $M$.
The spaces $L^p$ for $1< p < \infty$ are reflexive, meaning they are the duals
of their duals.  $L^1$ and $L^\infty$ are not reflexive.

The {\it Hahn-Banach} theorem says that whenever $\lambda$ is a bounded
linear functional on some vector subspace $V$ of Banach space $X$ with
bound $M$ on $V$, then there is a (non-unique) extension of $\lambda$ to all of $X$
with the same bound.

For two Banach spaces $X$ and $Y$, the Banach space $X \cap Y$ is the set intersection
normed by $\|x\|_{X\cap Y} = \|x\|_X + \|x\|_Y$.  The dual space is
$X^* + Y^*$, normed as $\|u\|_{X^*+Y^*} = \inf \{\|u_1\|_X + \|u_2\|_Y: u_1 + u_2 = u\}$.

{\bf Weak topologies.}
Apart from the norm topologies, there is another class of topologies on Banach spaces
which appear in this paper.  The weak topology on a Banach space $X$ is 
given by neighborhood bases of the origin of the form 
\[
U(0;\{\xi\},\epsilon) = \{ y\in X : | \langle \xi_i, y\rangle| < \epsilon:\, i=1,\ldots,n \},
\]
for some $\xi_1,\ldots, \xi_n \in X^*$ (a finite set) and $\epsilon > 0$.
Neighborhoods of other points in $X$ are obtained by translation.
Thus, a sequence $x_1,x_1,\ldots$ converges to $x$ in the weak topology if and
only if for any $\xi$ and $\epsilon$, there is some $N(\xi,\epsilon)$ such 
that $i > N(\xi,\epsilon)$ implies that $x_i \in U(x,\xi,\epsilon)$.
Reversing the roles of $X$ and $X^*$ gives the weak-$*$ topology on $X^*$.
The Banach-Alaoglu theorem, which is used a couple of times in this paper, says that
a norm-bounded, norm-closed subset $S$ of $X^*$ is compact in the weak-* topology.
If $X$ is separable, then the weak-* topology on $S$ is metrizable, so
weak-* compactness of $S$ is equivalent to sequential weak-* compactness.
If $X$ is reflexive, then the weak topology on $X$ is the same thing as the 
weak-* topology when viewing $X$ as $X^{**}$.  

\section{continuously differentiable convex functions in finite dimensions}
\label{cvx-appendix}

This Appendix contains a brief sketch of the convex analysis results
used in the proof of Lemma \ref{local-niceness}.
This material can be found in \S 6.7 of Ref. \cite{Florenzano},
and also in Ref. \cite{Rockafellar} with more work.  
It seems advisable to give a proof here because many books on 
convex analysis do not discuss it.

Here is the general setting in which we work:
$\{ f_n: n=1,2,\ldots \}$ and $f$ are convex functions on $B_{R} = B(0,R)$, 
the open ball of radius $R$ in ${\Bbb R}^n$, which are uniformly
bounded above and below.  (5) is the main result.

\noindent 1). The bounds, say $a < b$, imply that if ${\bm x} \in B_R$,
there is some open neighborhood of ${\bm x}$ on which 
$|f({\bm y}) - f({\bm y^\prime})|\le  c |{\bm y} - {\bm y^\prime}|$,
for $c$ which depends on the neighborhood.  This is called local Lipschitz continuity.

For, if ${\bm x} \in B_R$,
some open ball $B({\bm x},2\delta)$ around ${\bm x}$ is contained in $B_R$.
From the bounds and convexity of $f$ it is easy to see that
a chord to the graph of $f_n$ in $B({\bm x},\delta)$ then cannot exceed
$(b-a)/\delta$.

\noindent 2). The pointwise suprememum $\bigvee_{n} f_n$, and
the pointwise limit $\lim_{n} f_n({\bm x})$, if it exists,
are convex over $B_R$, since lack of convexity can be diagnosed from
only three points.  From these it follows that  
$\limsup_{n} f_n({\bm x})$, which always exists, is convex.

\noindent 3). If $f_n \to f$ pointwise, then the convergence is locally uniform:
given ${\bm x}$ and tolerance $\epsilon$, there is a neighborhood $U$
of ${\bm x}$ such that $|f_n - f|  < \epsilon$ over $U$, for
large enough $n$.  This follows from the fact that $\{ f_n \}$ and $f$
are equi-Lipschitz.

\noindent 4). Now, if $f_n \to f$ pointwise, then 
$\limsup f_n^\prime({\bm x};{\bm y}) \le f^\prime({\bm x};{\bm y})$
for any ${\bm y} \in {\Bbb R}^n$.  

This is essentially a one-dimensional problem,
and hinges on the fact that the difference quotients in the definition of
directional derivative converge monotonically:
With $\delta_\lambda f({\bm x};{\bm y}) \defeq \lambda^{-1}[f({\bm x}+\lambda {\bm y})-f({\bm y})]$ 
for $\lambda > 0$,
$\delta_\lambda f({\bm x};{\bm y}) \downarrow f^\prime({\bm x};{\bm y})$ as $\lambda \downarrow 0$. 
To prove the assertion, fix $\epsilon > 0$.  
Clearly, $\exists \lambda_0 > 0$ such that $\lambda \le \lambda_0$ implies that 
$\delta_\lambda f({\bm x};{\bm y}) < f^\prime({\bm x};{\bm y}) + \epsilon/2$.
Then, there is $N$ such that $n \ge N$ implies that both 
$|f_n({\bm x})-f({\bm x})|$ and $|f_n({\bm x}+\lambda_0 {\bm y})-f({\bm x}+\lambda_0 {\bm y})|$
are less than $\lambda_0 \epsilon/4$, so that
$\delta_{\lambda_0} f_n({\bm x};{\bm y}) < 
\delta_{\lambda_0} f({\bm x};{\bm y}) + \epsilon/2 < f^\prime({\bm x};{\bm y}) + \epsilon$.
Due to monotonicity of $\delta_\lambda f_n({\bm x};{\bm y})$, therefore,
$n\ge N_\epsilon$ implies that $f_n^\prime({\bm x};{\bm y}) < f^\prime({\bm x};{\bm y}) + \epsilon$,
and finally, $\limsup_n f_n^\prime({\bm x};{\bm y}) \le f^\prime({\bm x};{\bm y})$.

\noindent 5).  If $f_n \to f$ pointwise, then for ${\bm x} \in B_R$,
$\partial f_n^\prime({\bm x}) \subset \partial f({\bm x}) + \epsilon B_1$ for
large enough $n$.  If $f_n$ and $f$ are G\^ateaux differentiable, this
implies that $Df_n \to Df$ pointwise.

Since $f^\prime({\bm x};{\bm y})$ and $f_n^\prime({\bm x};{\bm y})$ are 
themselves convex functions of ${\bm y}$ for fixed ${\bm x}$, we can
apply (2) and (4) to see that $n \ge N({\bm y})$ implies 
$f_n^\prime({\bm x};{\bm y}) < f^\prime({\bm x};{\bm y}) + \epsilon$.
But, (3) shows that $N({\bm y})$ can be chosen independently of ${\bm y}$
in the unit sphere.

\noindent 6).  Applying (5) to $f_n({\bm x}) \defeq f({\bm x}-{\bm z}_n)$ 
with ${\bm z}_n\to 0$ shows that $f$ is $C^1$.

\section{\texorpdfstring{Genericity of $X^*$-nonrepresentable densities}{X*-nonrepresentability is generic}}
\label{bad-potentials}

\begin{prop}
\label{bad-neighbors}
Given $(\rho_0,v_0) \in \partial F$,
For any $\epsilon > 0$ and $\Lambda > 0$,
there is an $X^*$-representable density $\rho$ with
$\| \rho - \rho_0 \|_X < \epsilon$ and $|F[\rho] - F[\rho_0]| < \epsilon$,
such that $\|v[\rho] \rho \|_1 > \Lambda$, and therefore $\|v[\rho]\|_{X^*} > \Lambda$.
\end{prop}

\begin{proof}
Since Thm. \ref{V-limit} showed that any $X^*$-representable density is approximable
by $L^\infty$-representable densities, assume that $v_0 \in L^\infty$.

Define
\beq
{w}_\ell \defeq \eta({\bm x}) \sin \frac{x}{\ell}
= {\ell} \, \eta({\bm x}) \left(-\frac{\partial}{\partial x_1} \cos \frac{x_1}{\ell} \right),
\eeq
where $\eta({\bm x})$ is a smooth compactly supported bump function
($\eta \in C_0^\infty({\Bbb R}^3)$ and $0 \le \eta \le 1$).

The rough idea is that for very large $M$, if $\ell$ is small enough,
the ground state density for potential $v + M w_\ell$ cannot exploit $M w_\ell$
because doing so requires too much oscillation and therefore costs intrinsic energy
with the result that that ground state density is hardly different from $\rho_0$.
A slightly indirect approach and possibly a small additional modification of the 
potential is necessary to prove it.  The bound (\ref{unexploitable}) below and 
the Ekeland variational principle are key.

Integrating by parts, 
\beq
\left| \int {w}_\ell({\bm x}) \, \rho({\bm x}) d{\bm x} \right| 
 =
 {2\ell} \left| \int \cos \frac{x_1}{\ell} \, 
\left( \sqrt{{\rho}{\eta}} \frac{\partial \sqrt{{\rho}\eta}}{\partial x_1}
\right) d{\bm x} \right|.
\nonumber
\eeq
Applying the Cauchy-Schwarz inequality and using that $\cos^2 {x_1}/{\ell} \le 1$ yields,
\beq
|\langle {w}_\ell , \rho \rangle | \le 
{2\ell} \left( \int \eta \rho \, d{\bm x} \right)^{1/2}
\left[ 
\left( \int \eta | \nabla \rho^{1/2}|^2 \, d{\bm x} \right)^{1/2}
\right.
 + 
\left.
\left( \int \rho | \nabla \eta^{1/2}|^2 \, d{\bm x} \right)^{1/2}
\right].
\nonumber
\eeq
Finally, using inequality (\ref{Lieb-sandwich}), and overestimating the first
and last integrals,
\beq
|\langle w_\ell , \rho \rangle |  \le 
\ell \left( c_1 F[\rho]^{1/2} + c_2 \right),
\label{unexploitable}
\eeq
for some positive constants $c_1$ and $c_2$ which depend on total particle number.  

First, assuming $\ell$ small enough (depending on $\rho_0$) that
$ \| M w_\ell \rho_0 \|_{1} \ge \frac{M}{3} \int \rho_0 \eta \, d{\bm x} =: c_3 {M}$,
fix $M > 1$ large enough that (uniformly for small $\ell$)
\beq
\| (M w_\ell + v_0) \rho_0 \|_{1} > 2\Lambda,
\label{Mwv}
\eeq
and without loss of generality assume that $\epsilon < 1/M < 1$.
A few more upper bounds on $\ell$ will be imposed.

From now on, we will be interested only in densities $\rho$ satisfying the condition
\beq
M \langle w_\ell,\rho \rangle \le 
F[\rho] + \langle v_0 + Mw_\ell,\rho \rangle \le F[\rho_0] + \langle v_0 + Mw_\ell,\rho_0 \rangle
= M \langle w_\ell,\rho_0 \rangle
\label{better-than-rho0}
\eeq
Define $\tilde{\epsilon}$ by
\beq
{\tilde{\epsilon}}^{1/2} \defeq 
\epsilon \cdot \min \left\{ 1,\|v_0\|_{\infty}^{-1},\|\rho_0\|_X^{-1} \right\}
\le \epsilon.
\label{epsilon-tilde}
\eeq
Claim: if $\ell$ is small enough, then
$M\langle w_\ell, \rho_0 \rangle < {\tilde{\epsilon}}/3$,
and for any $\rho$ satisfying (\ref{better-than-rho0}), 
$M\langle w_\ell, \rho \rangle > {\tilde{\epsilon}}/3$.
The first part is simple since $\langle w_\ell, \rho_0 \rangle \to  0$ as $\ell \to 0$.
Given that, the second follows by use of the bound (\ref{unexploitable}).
Essentially the problem involves a quadratic in $F[\rho]^{1/2}$.
If $M \ell$ is very small, violation of $M\langle w_\ell, \rho \rangle > {\tilde{\epsilon}}/3$
would require such a large $F[\rho]^{1/2}$ that the upper bound in (\ref{better-than-rho0})
would also be violated.  We will not write down explicit criteria, but assume $\ell$
has been so constrained, with the result that 
$F[\rho_0] + \langle v_0 + Mw_\ell, \rho_0\rangle \le E[v_0 + Mw_\ell] + {\tilde{\epsilon}}$.

The ground is thus prepared for an application of the Ekeland variational principle.
It guarantees existence of a ground-state density $\rho$ of $v \defeq v_0 + \Delta v+ Mw_\ell$ 
where the perturbation $\Delta v$ satisfies $\|\Delta v\|_{X^*} < {\tilde{\epsilon}}^{1/2}$,
and $\rho$ satisfies $\|\rho - \rho_0\|_{X} < {\tilde{\epsilon}}^{1/2}$ and (\ref{better-than-rho0}).
Thus,
\beqa
F[\rho] 
& \le &
F[\rho_0] + \langle v_0, (\rho_0-\rho) \rangle + \langle M w_\ell, (\rho_0-\rho) \rangle
+ \langle \Delta v, (\rho_0-\rho) \rangle
\nonumber \\
& \le  &
F[\rho_0] + \epsilon + \tilde{\epsilon} + \tilde{\epsilon}^{1/2}
\le F[\rho_0] + 3 \epsilon.
\nonumber
\eeqa
This shows that we get $\rho$ close to $\rho_0$ in $L^1$ norm and intrinsic
energy as claimed in the statement of the Theorem.  
The upper bound on $F[\rho]$ suffices due to lower semicontinuity.

All that remains to be checked are the sizes of 
$\|v \rho\|_{1}$ and $\|v\|_{X^*}$.
Calculate using (\ref{Mwv}) and (\ref{epsilon-tilde}):
\beqa
\|(v_0+\Delta v + Mw_\ell)\rho\|_{1} &> &
\|(v_0+ Mw_\ell)\rho_0\|_{1} - 
\|(v_0+ Mw_\ell)(\rho - \rho_0)\|_{1} - 
\|\Delta v \, \rho\|_{1} 
\nonumber \\
& > & 2\Lambda - (M+\|v_0\|_{\infty}) {\tilde\epsilon}^{1/2} - 
\|\Delta v\|_{X^*} \|\rho_0\|_{X} - \|\Delta v\|_{X^*} \|\rho - \rho_0\|_{X}
\nonumber \\
& > & 2\Lambda - (1 + \epsilon) - \epsilon - \epsilon > 2\Lambda - 4.
\nonumber
\eeqa
Finally, notice that $\|v\|_{X^*} \ge \|v \rho \|_1/\|\rho\|_{X}$ and
that $\|\rho\|_{X} \ge 1$.
\end{proof}

This proposition will be used here as a springboard toward the conclusion that
failure of $X^*$-representability is topologically generic in a nontrivial sense.
Topological genericity is meant here in the sense of the Baire Category Theorem.
Recall that in a complete metric space $S$, a subset is {\it generic} if it is a countable 
intersection of open dense sets.  A subset is {\it meager} if it is a countable
union of nowhere dense (empty interior) closed sets.  The complement of a generic
set is meager.  Obviously, a countable intersection of generic subsets is itself
generic.  Though this notion of `generic' is not without flaws, it is commonly 
accepted as a criterion of just what the name implies.  At any rate, a generic
set is dense.
To apply the concept, we need a reasonable metric which makes $\JN$ a generic set in
its completion.
Thus, equip $\JN$ with the metric 
\beq
d(\rho,\rho^\prime) = \|\rho-\rho^\prime\|_{X} + |F[\rho]-F[\rho^\prime]|.
\eeq
This metric is actually just the graph norm of $F$ as $\graph F \subset X \times {\Bbb R}$.
an observation which is key to identifying the completion of $\JN$ with respect to $d$
with $\epi F \subset X\times{\Bbb R}$.  Pictorially,
$$
\begin{CD}
(\JN,d) @>{\subset}>> (\overline{\JN},d) \\
@VV{\simeq}V  @VV{\simeq}V \\
\graph F @>{\subset}>>  \epi F 
\end{CD}
$$
The identification is established by showing that $\epi F$ is a 
closed subset of $X\times {\Bbb R}$ (hence a $d$-complete space),
and that $\JN$ is dense in $\epi F$.
The first is a simple consequence of lower semicontinuity of $F$.
The second of the fact that $F$ is unbounded above on every 
$X$-neighborhood of any point in $\JN$:  
take $(\rho,z) \in \epi F$,  $z > F[\rho]$, and $\epsilon > 0$.
Then, there is $\rho^\prime$ satisfying
$\| \rho^\prime - \rho\|_X < \epsilon$ and $F[\rho^\prime] > z$.
Over the line segment $[\rho,\rho^\prime]$, $F$ goes from $F[\rho] < z$
to $F[\rho^\prime] > z$, hence must pass through $z$.  This shows that
there is a point $(\rho^{\prime\prime},z) \in \graph F$ within 
$d$-distance $\epsilon$ of $(\rho,z)$.  Hence, $\graph F$ is $d$-dense
in $\epi F$. 
In fact, $\JN$ is actually generic in $\overline{\JN}$.  For, since 
$\epi F$ is closed, so is $\epi (F + a)$ for $a > 0$, since this is just 
a translate of $\epi F$ in $X\times {\Bbb R}$.  The argument in the 
previous paragraph shows that $\epi (F+a)$, as a subset of $(\epi F,d)$, has empty interior.  
Therefore $\cup_{n\ge 1} \epi (F+1/n)$ is meager in $\overline{\JN}$.  But this last is {\em all} 
of $\epi F$ except for $\graph F$.  
This result is crucial to the endeavor.  A property $P$ defined on $\overline{\JN}$ is 
generic if and only if $P$-and-in-$\JN$ is generic, so that the status of
elements in $\overline{\JN}\setminus{\JN}$ is irrelevant.

Paraphrased in terms of $d$, the Proposition says that neither $\|v[\rho]\rho\|_1$ nor
$\|v[\rho]\|_{X^*}$ (without loss set equal to $+\infty$ outside the effective domain)
is upper semicontinuous with respect to $d$. 
Denote by $A_\alpha$ the set 
$A_\alpha \defeq \{ \rho\in\JN: \| v[\rho] \|_{X^*} \le \alpha \}$,
of densities with representing potentials having $X^*$ norms 
not exceeding $\alpha$.  According to Section \ref{fine-grained}, $A_\alpha$ is closed 
with respect to $X$ norm, {\it a fortiori} with respect to $d$, due to
lower semicontinuity of $\|v[\cdot]\|_{X^*}$.  
On the other hand, the Proposition shows that $A_\alpha$ is nowhere dense.
Consequently, ${\mathcal R} = \cup_{\alpha=1}^\infty {A}_\alpha$
is meager, being a countable union of closed nowhere dense sets.
Since ${\mathcal R}$ contains every $X^*$-representable density,
not being $X^*$-representable is generic.


\begin{thebibliography}{10}%
\makeatletter
\providecommand \@ifxundefined [1]{%
 \ifx #1\undefined \expandafter \@firstoftwo
 \else \expandafter \@secondoftwo
\fi
}%
\providecommand \@ifnum [1]{%
 \ifnum #1\expandafter \@firstoftwo
 \else \expandafter \@secondoftwo
\fi
}%
\providecommand \enquote [1]{``#1''}%
\providecommand \bibnamefont  [1]{#1}%
\providecommand \bibfnamefont [1]{#1}%
\providecommand \citenamefont [1]{#1}%
\providecommand\href[0]{\@sanitize\@href}%
\providecommand\@href[1]{\endgroup\@@startlink{#1}\endgroup\@@href}%
\providecommand\@@href[1]{#1\@@endlink}%
\providecommand \@sanitize [0]{\begingroup\catcode`\&12\catcode`\#12\relax}%
\@ifxundefined \pdfoutput {\@firstoftwo}{%
 \@ifnum{\z@=\pdfoutput}{\@firstoftwo}{\@secondoftwo}%
}{%
 \providecommand\@@startlink[1]{\leavevmode}%
 \providecommand\@@endlink[0]{}%
}{%
 \providecommand\@@startlink[1]{%
  \leavevmode
  \pdfstartlink
   attr{/Border[0 0 1 ]/H/I/C[0 1 1]}%
   user{/Subtype/Link/A<</Type/Action/S/URI/URI(#1)>>}%
  \relax
 }%
 \providecommand\@@endlink[0]{\pdfendlink}%
}%
\providecommand \url  [0]{\begingroup\@sanitize \@url }%
\providecommand \@url [1]{\endgroup\@href {#1}{\urlprefix}}%
\providecommand \urlprefix [0]{URL }%
\providecommand \Eprint[0]{\href }%
\@ifxundefined \urlstyle {%
  \providecommand \doi [1]{doi:\discretionary{}{}{}#1}%
}{%
  \providecommand \doi [0]{doi:\discretionary{}{}{}\begingroup
  \urlstyle{rm}\Url }%
}%
\providecommand \doibase [0]{http://dx.doi.org/}%
\providecommand \Doi[1]{\href{\doibase#1}}%
\providecommand \bibAnnote [3]{%
  \BibitemShut{#1}%
  \begin{quotation}\noindent
    \textsc{Key:}\ #2\\\textsc{Annotation:}\ #3%
  \end{quotation}%
}%
\providecommand \bibAnnoteFile [2]{%
  \IfFileExists{#2}{\bibAnnote {#1} {#2} {\input{#2}}}{}%
}%
\providecommand \typeout [0]{\immediate \write \m@ne }%
\providecommand \selectlanguage [0]{\@gobble}%
\providecommand \bibinfo [0]{\@secondoftwo}%
\providecommand \bibfield [0]{\@secondoftwo}%
\providecommand \translation [1]{[#1]}%
\providecommand \BibitemOpen[0]{}%
\providecommand \bibitemStop [0]{}%
\providecommand \bibitemNoStop [0]{.\EOS\space}%
\providecommand \EOS [0]{\spacefactor3000\relax}%
\providecommand \BibitemShut [1]{\csname bibitem#1\endcsname}%
\bibitem{Parr-Yang}%
  \BibitemOpen
  \bibfield{author}{%
  \bibinfo {author} {\bibfnamefont{R.}~\bibnamefont{Parr}}\ and\ \bibinfo
  {author} {\bibfnamefont{W.}~\bibnamefont{Yang}},\ }%
  \emph{\bibinfo {title} {Density-Functional Theory of Atoms and Molecules}}\
  (\bibinfo {publisher} {Clarendon},\ \bibinfo {address} {Cambridge},\ \bibinfo
  {year} {1989})%
  \bibAnnoteFile{NoStop}{Parr-Yang}%
\bibitem{Dreizler-Gross}%
  \BibitemOpen
  \bibfield{author}{%
  \bibinfo {author} {\bibfnamefont{R.~M.}\ \bibnamefont{Dreizler}}\ and\
  \bibinfo {author} {\bibfnamefont{E.~K.~U.}\ \bibnamefont{Gross}},\ }%
  \emph{\bibinfo {title} {Density Functional Theory: an approach to the quantum
  many-body problem}}\ (\bibinfo {publisher} {Springer-Verlag},\ \bibinfo
  {address} {Berlin},\ \bibinfo {year} {1990})%
  \bibAnnoteFile{NoStop}{Dreizler-Gross}%
\bibitem{Eschrig}%
  \BibitemOpen
  \bibfield{author}{%
  \bibinfo {author} {\bibfnamefont{H.}~\bibnamefont{Eschrig}},\ }%
  \emph{\bibinfo {title} {The Fundamentals of Density Functional Theory}}\
  (\bibinfo {publisher} {Teubner},\ \bibinfo {address} {Stuttgart, Leipzig},\
  \bibinfo {year} {1996})%
  \bibAnnoteFile{NoStop}{Eschrig}%
\bibitem{Martin}%
  \BibitemOpen
  \bibfield{author}{%
  \bibinfo {author} {\bibfnamefont{R.~M.}\ \bibnamefont{Martin}},\ }%
  \emph{\bibinfo {title} {Electronic Structure: Basic Theory and Practical
  Methods}}\ (\bibinfo {publisher} {Cambridge University Press},\ \bibinfo
  {address} {Cambridge,New York},\ \bibinfo {year} {2004})%
  \bibAnnoteFile{NoStop}{Martin}%
\bibitem{Capelle06}%
  \BibitemOpen
  \bibfield{author}{%
  \bibinfo {author} {\bibfnamefont{K.}~\bibnamefont{Capelle}},\ }%
  \bibfield{journal}{%
  \bibinfo {journal} {Braz J Phys}\ }%
  \textbf{\bibinfo {volume} {36}},\ \bibinfo {pages} {1318} (\bibinfo {year}
  {2006}),\ \bibinfo {note} {arXiv:cond-mat/0211443}%
  \bibAnnoteFile{NoStop}{Capelle06}%
\bibitem{KS65}%
  \BibitemOpen
  \bibfield{author}{%
  \bibinfo {author} {\bibfnamefont{W.}~\bibnamefont{Kohn}}\ and\ \bibinfo
  {author} {\bibfnamefont{L.~J.}\ \bibnamefont{Sham}},\ }%
  \bibfield{journal}{%
  \bibinfo {journal} {Phys Rev}\ }%
  \textbf{\bibinfo {volume} {140}},\ \bibinfo {pages} {A1333} (\bibinfo {year}
  {1965})%
  \bibAnnoteFile{NoStop}{KS65}%
\bibitem{Ho08}%
  \BibitemOpen
  \bibfield{author}{%
  \bibinfo {author} {\bibfnamefont{G.~S.}\ \bibnamefont{Ho}}, \bibinfo {author}
  {\bibfnamefont{V.~L.}\ \bibnamefont{Lign\`eres}},\ and\ \bibinfo {author}
  {\bibfnamefont{E.~A.}\ \bibnamefont{Carter}},\ }%
  \bibfield{journal}{%
  \bibinfo {journal} {Phys Rev B}\ }%
  \textbf{\bibinfo {volume} {78}},\ \bibinfo {pages} {045105} (\bibinfo {year}
  {2008})%
  \bibAnnoteFile{NoStop}{Ho08}%
\bibitem{Chai07}%
  \BibitemOpen
  \bibfield{author}{%
  \bibinfo {author} {\bibfnamefont{J.-D.}\ \bibnamefont{Chai}}\ and\ \bibinfo
  {author} {\bibfnamefont{J.~D.}\ \bibnamefont{Weeks}},\ }%
  \bibfield{journal}{%
  \bibinfo {journal} {Phys Rev B}\ }%
  \textbf{\bibinfo {volume} {75}},\ \bibinfo {pages} {205122} (\bibinfo {year}
  {2007})%
  \bibAnnoteFile{NoStop}{Chai07}%
\bibitem{Garcia-Cervera07}%
  \BibitemOpen
  \bibfield{author}{%
  \bibinfo {author} {\bibfnamefont{C.~J.}\ \bibnamefont{Garc\'ia-Cervera}},\ }%
  \bibfield{journal}{%
  \bibinfo {journal} {Comm Comp Phys}\ }%
  \textbf{\bibinfo {volume} {2}},\ \bibinfo {pages} {334} (\bibinfo {year}
  {2007})%
  \bibAnnoteFile{NoStop}{Garcia-Cervera07}%
\bibitem{Levy79}%
  \BibitemOpen
  \bibfield{author}{%
  \bibinfo {author} {\bibfnamefont{M.}~\bibnamefont{Levy}},\ }%
  \bibfield{journal}{%
  \bibinfo {journal} {Proc Natl Acad Sci USA}\ }%
  \textbf{\bibinfo {volume} {76}},\ \bibinfo {pages} {6062} (\bibinfo {year}
  {1979})%
  \bibAnnoteFile{NoStop}{Levy79}%
\bibitem{Levy82}%
  \BibitemOpen
  \bibfield{author}{%
  \bibinfo {author} {\bibfnamefont{M.}~\bibnamefont{Levy}},\ }%
  \bibfield{journal}{%
  \bibinfo {journal} {Phys Rev A}\ }%
  \textbf{\bibinfo {volume} {26}},\ \bibinfo {pages} {1200} (\bibinfo {year}
  {1982})%
  \bibAnnoteFile{NoStop}{Levy82}%
\bibitem{Valone80}%
  \BibitemOpen
  \bibfield{author}{%
  \bibinfo {author} {\bibfnamefont{S.}~\bibnamefont{Valone}},\ }%
  \bibfield{journal}{%
  \bibinfo {journal} {J Chem Phys}\ }%
  \textbf{\bibinfo {volume} {73}},\ \bibinfo {pages} {4653} (\bibinfo {year}
  {1980})%
  \bibAnnoteFile{NoStop}{Valone80}%
\bibitem{Lieb83}%
  \BibitemOpen
  \bibfield{author}{%
  \bibinfo {author} {\bibfnamefont{E.~H.}\ \bibnamefont{Lieb}},\ }%
  \bibfield{journal}{%
  \bibinfo {journal} {Int J Quantum Chem}\ }%
  \textbf{\bibinfo {volume} {24}},\ \bibinfo {pages} {243} (\bibinfo {year}
  {1983})%
  \bibAnnoteFile{NoStop}{Lieb83}%
\bibitem{Argaman02}%
  \BibitemOpen
  \bibfield{author}{%
  \bibinfo {author} {\bibfnamefont{N.}~\bibnamefont{Argaman}}\ and\ \bibinfo
  {author} {\bibfnamefont{G.}~\bibnamefont{Makov}},\ }%
  \bibfield{journal}{%
  \bibinfo {journal} {Phys Rev B}\ }%
  \textbf{\bibinfo {volume} {66}},\ \bibinfo {pages} {052413} (\bibinfo {year}
  {2002})%
  \bibAnnoteFile{NoStop}{Argaman02}%
\bibitem{CCR85}%
  \BibitemOpen
  \bibfield{author}{%
  \bibinfo {author} {\bibfnamefont{J.~T.}\ \bibnamefont{Chayes}}, \bibinfo
  {author} {\bibfnamefont{L.}~\bibnamefont{Chayes}},\ and\ \bibinfo {author}
  {\bibfnamefont{M.~B.}\ \bibnamefont{Ruskai}},\ }%
  \bibfield{journal}{%
  \bibinfo {journal} {J Stat Phys}\ }%
  \textbf{\bibinfo {volume} {38}},\ \bibinfo {pages} {497} (\bibinfo {year}
  {1985})%
  \bibAnnoteFile{NoStop}{CCR85}%
\bibitem{Kohn83}%
  \BibitemOpen
  \bibfield{author}{%
  \bibinfo {author} {\bibfnamefont{W.}~\bibnamefont{Kohn}},\ }%
  \bibfield{journal}{%
  \bibinfo {journal} {Phys Rev Lett}\ }%
  \textbf{\bibinfo {volume} {51}},\ \bibinfo {pages} {1596} (\bibinfo {year}
  {1983})%
  \bibAnnoteFile{NoStop}{Kohn83}%
\bibitem{Ullrich05}%
  \BibitemOpen
  \bibfield{author}{%
  \bibinfo {author} {\bibfnamefont{C.~A.}\ \bibnamefont{Ullrich}},\ }%
  \bibfield{journal}{%
  \bibinfo {journal} {Phys Rev B}\ }%
  \textbf{\bibinfo {volume} {72}},\ \bibinfo {pages} {073102} (\bibinfo {year}
  {2005})%
  \bibAnnoteFile{NoStop}{Ullrich05}%
\bibitem{Harriman-GDM1}%
  \BibitemOpen
  \bibfield{author}{%
  \bibinfo {author} {\bibfnamefont{J.~E.}\ \bibnamefont{Harriman}},\ }%
  \bibfield{journal}{%
  \bibinfo {journal} {Phys Rev A}\ }%
  \textbf{\bibinfo {volume} {17}},\ \bibinfo {pages} {1249} (\bibinfo {year}
  {1978})%
  \bibAnnoteFile{NoStop}{Harriman-GDM1}%
\bibitem{Harriman-GDM2}%
  \BibitemOpen
  \bibfield{author}{%
  \bibinfo {author} {\bibfnamefont{J.~E.}\ \bibnamefont{Harriman}},\ }%
  \bibfield{journal}{%
  \bibinfo {journal} {Phys Rev A}\ }%
  \textbf{\bibinfo {volume} {17}},\ \bibinfo {pages} {1257} (\bibinfo {year}
  {1978})%
  \bibAnnoteFile{NoStop}{Harriman-GDM2}%
\bibitem{Harriman-GDM3}%
  \BibitemOpen
  \bibfield{author}{%
  \bibinfo {author} {\bibfnamefont{J.~E.}\ \bibnamefont{Harriman}},\ }%
  \bibfield{journal}{%
  \bibinfo {journal} {Intl J Quantum Chem}\ }%
  \textbf{\bibinfo {volume} {15}},\ \bibinfo {pages} {611} (\bibinfo {year}
  {1979})%
  \bibAnnoteFile{NoStop}{Harriman-GDM3}%
\bibitem{Harriman-GDM4}%
  \BibitemOpen
  \bibfield{author}{%
  \bibinfo {author} {\bibfnamefont{J.~E.}\ \bibnamefont{Harriman}},\ }%
  \bibfield{journal}{%
  \bibinfo {journal} {Phys Rev A}\ }%
  \textbf{\bibinfo {volume} {27}},\ \bibinfo {pages} {632} (\bibinfo {year}
  {1983})%
  \bibAnnoteFile{NoStop}{Harriman-GDM4}%
\bibitem{Lammert06}%
  \BibitemOpen
  \bibfield{author}{%
  \bibinfo {author} {\bibfnamefont{P.~E.}\ \bibnamefont{Lammert}},\ }%
  \bibfield{journal}{%
  \bibinfo {journal} {J Chem Phys}\ }%
  \textbf{\bibinfo {volume} {125}},\ \bibinfo {pages} {074114} (\bibinfo {year}
  {2006})%
  \bibAnnoteFile{NoStop}{Lammert06}%
\bibitem{Gunnarsson-Lundqvist76}%
  \BibitemOpen
  \bibfield{author}{%
  \bibinfo {author} {\bibfnamefont{O.}~\bibnamefont{Gunnarsson}}\ and\ \bibinfo
  {author} {\bibfnamefont{B.~I.}\ \bibnamefont{Lundqvist}},\ }%
  \bibfield{journal}{%
  \bibinfo {journal} {Phys Rev B}\ }%
  \textbf{\bibinfo {volume} {13}},\ \bibinfo {pages} {4274} (\bibinfo {year}
  {1976})%
  \bibAnnoteFile{NoStop}{Gunnarsson-Lundqvist76}%
\bibitem{Gidopoulos07}%
  \BibitemOpen
  \bibfield{author}{%
  \bibinfo {author} {\bibfnamefont{N.~I.}\ \bibnamefont{Gidopoulos}},\ }%
  \bibfield{journal}{%
  \bibinfo {journal} {Phys Rev B}\ }%
  \textbf{\bibinfo {volume} {75}},\ \bibinfo {pages} {134408} (\bibinfo {year}
  {2007})%
  \bibAnnoteFile{NoStop}{Gidopoulos07}%
\bibitem{vanLeeuwen03}%
  \BibitemOpen
  \bibfield{author}{%
  \bibinfo {author} {\bibfnamefont{R.}~\bibnamefont{van Leeuwen}},\ }%
  in\ \emph{\bibinfo {booktitle} {Advances in Quantum Chemistry}},\
  Vol.~\bibinfo {volume} {43},\ \bibinfo {editor} {edited by\ \bibinfo {editor}
  {\bibfnamefont{S.~J.}\ \bibnamefont{R.}}\ and\ \bibinfo {editor}
  {\bibfnamefont{B.}~\bibnamefont{E.}}}\ (\bibinfo {publisher} {Elsevier},\
  \bibinfo {address} {Amsterdam},\ \bibinfo {year} {2003})\ pp.\ \bibinfo
  {pages} {25--94}%
  \bibAnnoteFile{NoStop}{vanLeeuwen03}%
\bibitem{ET}%
  \BibitemOpen
  \bibfield{author}{%
  \bibinfo {author} {\bibfnamefont{I.}~\bibnamefont{Ekeland}}\ and\ \bibinfo
  {author} {\bibfnamefont{R.}~\bibnamefont{T\'emam}},\ }%
  \emph{\bibinfo {title} {Convex Analysis and Variational Problems}}\ (\bibinfo
  {publisher} {North-Holland},\ \bibinfo {address} {Amsterdam},\ \bibinfo
  {year} {1976})\ \bibinfo {note} {reprinted 1999 (SIAM, Philadelphia)}%
  \bibAnnoteFile{NoStop}{ET}%
\bibitem{Aubin-Ekeland}%
  \BibitemOpen
  \bibfield{author}{%
  \bibinfo {author} {\bibfnamefont{J.-P.}\ \bibnamefont{Aubin}}\ and\ \bibinfo
  {author} {\bibfnamefont{I.}~\bibnamefont{Ekeland}},\ }%
  \emph{\bibinfo {title} {Applied Nonlinear Functional Analysis}}\ (\bibinfo
  {publisher} {Wiley},\ \bibinfo {address} {New York},\ \bibinfo {year}
  {1984})\ \bibinfo {note} {reprinted (Dover, Mineola, NY, 2006)}%
  \bibAnnoteFile{NoStop}{Aubin-Ekeland}%
\bibitem{Borwein-Zhu}%
  \BibitemOpen
  \bibfield{author}{%
  \bibinfo {author} {\bibfnamefont{J.~M.}\ \bibnamefont{Borwein}}\ and\
  \bibinfo {author} {\bibfnamefont{Q.~J.}\ \bibnamefont{Zhu}},\ }%
  \emph{\bibinfo {title} {Techniques of Variational Analysis}}\ (\bibinfo
  {publisher} {Springer-Verlag},\ \bibinfo {address} {New York},\ \bibinfo
  {year} {2005})%
  \bibAnnoteFile{NoStop}{Borwein-Zhu}%
\bibitem{Phelps88}%
  \BibitemOpen
  \bibfield{author}{%
  \bibinfo {author} {\bibfnamefont{R.~R.}\ \bibnamefont{Phelps}},\ }%
  \emph{\bibinfo {title} {Convex functions, monotone operators and
  differentiability}},\ \bibinfo {edition} {2nd}\ ed.,\ \bibinfo {series}
  {Lecture Notes in Mathematics}, Vol.\ \bibinfo {volume} {1364}\ (\bibinfo
  {publisher} {Springer-Verlag},\ \bibinfo {address} {Berlin,New York},\
  \bibinfo {year} {1988})%
  \bibAnnoteFile{NoStop}{Phelps88}%
\bibitem{Harriman81}%
  \BibitemOpen
  \bibfield{author}{%
  \bibinfo {author} {\bibfnamefont{J.~E.}\ \bibnamefont{Harriman}},\ }%
  \bibfield{journal}{%
  \bibinfo {journal} {Phys Rev A}\ }%
  \textbf{\bibinfo {volume} {24}},\ \bibinfo {pages} {680} (\bibinfo {year}
  {1981})%
  \bibAnnoteFile{NoStop}{Harriman81}%
\bibitem{Zumbach-Maschke}%
  \BibitemOpen
  \bibfield{author}{%
  \bibinfo {author} {\bibfnamefont{G.}~\bibnamefont{Zumbach}}\ and\ \bibinfo
  {author} {\bibfnamefont{K.}~\bibnamefont{Maschke}},\ }%
  \bibfield{journal}{%
  \bibinfo {journal} {Phys Rev A}\ }%
  \textbf{\bibinfo {volume} {28}},\ \bibinfo {pages} {544} (\bibinfo {year}
  {1983})%
  \bibAnnoteFile{NoStop}{Zumbach-Maschke}%
\bibitem{Zumbach85}%
  \BibitemOpen
  \bibfield{author}{%
  \bibinfo {author} {\bibfnamefont{G.}~\bibnamefont{Zumbach}},\ }%
  \bibfield{journal}{%
  \bibinfo {journal} {Phys Rev A}\ }%
  \textbf{\bibinfo {volume} {31}},\ \bibinfo {pages} {1922} (\bibinfo {year}
  {1985})%
  \bibAnnoteFile{NoStop}{Zumbach85}%
\bibitem{EE84}%
  \BibitemOpen
  \bibfield{author}{%
  \bibinfo {author} {\bibfnamefont{H.}~\bibnamefont{Englisch}}\ and\ \bibinfo
  {author} {\bibfnamefont{R.}~\bibnamefont{Englisch}},\ }%
  \bibfield{journal}{%
  \bibinfo {journal} {Phys Status Solidi B}\ }%
  \textbf{\bibinfo {volume} {124}},\ \bibinfo {pages} {373} (\bibinfo {year}
  {1984})%
  \bibAnnoteFile{NoStop}{EE84}%
\bibitem{Ekeland79}%
  \BibitemOpen
  \bibfield{author}{%
  \bibinfo {author} {\bibfnamefont{I.}~\bibnamefont{Ekeland}},\ }%
  \bibfield{journal}{%
  \bibinfo {journal} {Bulletin of the American Mathematical Society}\ }%
  \textbf{\bibinfo {volume} {1}},\ \bibinfo {pages} {443} (\bibinfo {year}
  {1979})%
  \bibAnnoteFile{NoStop}{Ekeland79}%
\bibitem{EE83}%
  \BibitemOpen
  \bibfield{author}{%
  \bibinfo {author} {\bibfnamefont{H.}~\bibnamefont{Englisch}}\ and\ \bibinfo
  {author} {\bibfnamefont{R.}~\bibnamefont{Englisch}},\ }%
  \bibfield{journal}{%
  \bibinfo {journal} {Physica A}\ }%
  \textbf{\bibinfo {volume} {121}},\ \bibinfo {pages} {253} (\bibinfo {year}
  {1983})%
  \bibAnnoteFile{NoStop}{EE83}%
\bibitem{RSIV}%
  \BibitemOpen
  \bibfield{author}{%
  \bibinfo {author} {\bibfnamefont{M.}~\bibnamefont{Reed}}\ and\ \bibinfo
  {author} {\bibfnamefont{B.}~\bibnamefont{Simon}},\ }%
  \emph{\bibinfo {title} {Methods of Modern Mathematical Physics, Vol. IV,
  Analysis of Operators}}\ (\bibinfo {publisher} {Academic Press},\ \bibinfo
  {address} {New York},\ \bibinfo {year} {1978})%
  \bibAnnoteFile{NoStop}{RSIV}%
\bibitem{Lieb-Loss}%
  \BibitemOpen
  \bibfield{author}{%
  \bibinfo {author} {\bibfnamefont{E.~H.}\ \bibnamefont{Lieb}}\ and\ \bibinfo
  {author} {\bibfnamefont{M.}~\bibnamefont{Loss}},\ }%
  \emph{\bibinfo {title} {Analysis}}\ (\bibinfo {publisher} {American
  Mathematical Society},\ \bibinfo {address} {Providence, R.I.},\ \bibinfo
  {year} {1997})%
  \bibAnnoteFile{NoStop}{Lieb-Loss}%
\bibitem{Deimling}%
  \BibitemOpen
  \bibfield{author}{%
  \bibinfo {author} {\bibfnamefont{K.}~\bibnamefont{Deimling}},\ }%
  \emph{\bibinfo {title} {Nonlinear Functional Analysis}}\ (\bibinfo
  {publisher} {Springer-Verlag},\ \bibinfo {address} {Berlin},\ \bibinfo {year}
  {1985})\ \bibinfo {note} {reprinted (Dover, Mineola NY, 2010)}%
  \bibAnnoteFile{NoStop}{Deimling}%
\bibitem{Aubin-Frankowska}%
  \BibitemOpen
  \bibfield{author}{%
  \bibinfo {author} {\bibfnamefont{J.-P.}\ \bibnamefont{Aubin}}\ and\ \bibinfo
  {author} {\bibfnamefont{H.}~\bibnamefont{Frankowska}},\ }%
  \emph{\bibinfo {title} {Set-Valued Analysis}}\ (\bibinfo {publisher}
  {Birkh\"auser},\ \bibinfo {address} {Boston},\ \bibinfo {year} {2009})%
  \bibAnnoteFile{NoStop}{Aubin-Frankowska}%
\bibitem{Adams}%
  \BibitemOpen
  \bibfield{author}{%
  \bibinfo {author} {\bibfnamefont{R.~A.}\ \bibnamefont{Adams}},\ }%
  \emph{\bibinfo {title} {Sobolev Spaces}}\ (\bibinfo {publisher} {Academic
  Press},\ \bibinfo {address} {New York},\ \bibinfo {year} {1975})%
  \bibAnnoteFile{NoStop}{Adams}%
\bibitem{Attouch}%
  \BibitemOpen
  \bibfield{author}{%
  \bibinfo {author} {\bibfnamefont{H.}~\bibnamefont{Attouch}}, \bibinfo
  {author} {\bibfnamefont{B.}~\bibnamefont{Buttazzo}},\ and\ \bibinfo {author}
  {\bibfnamefont{G.}~\bibnamefont{Michaille}},\ }%
  \emph{\bibinfo {title} {Variational Analysis in Sobolev and BV Spaces}}\
  (\bibinfo {publisher} {SIAM},\ \bibinfo {address} {Philadelphia},\ \bibinfo
  {year} {2006})%
  \bibAnnoteFile{NoStop}{Attouch}%
\bibitem{Acosta-Duran03}%
  \BibitemOpen
  \bibfield{author}{%
  \bibinfo {author} {\bibfnamefont{G.}~\bibnamefont{Acosta}}\ and\ \bibinfo
  {author} {\bibfnamefont{R.~G.}\ \bibnamefont{Dur\'an}},\ }%
  \bibfield{journal}{%
  \bibinfo {journal} {Proc Am Math Soc}\ }%
  \textbf{\bibinfo {volume} {132}},\ \bibinfo {pages} {195} (\bibinfo {year}
  {2003})%
  \bibAnnoteFile{NoStop}{Acosta-Duran03}%
\bibitem{Bebendorf03}%
  \BibitemOpen
  \bibfield{author}{%
  \bibinfo {author} {\bibfnamefont{M.}~\bibnamefont{Bebendorf}},\ }%
  \bibfield{journal}{%
  \bibinfo {journal} {Zeitschrift f\"ur Analysis und ihre Anwendungen}\ }%
  \textbf{\bibinfo {volume} {22}},\ \bibinfo {pages} {751} (\bibinfo {year}
  {2003})%
  \bibAnnoteFile{NoStop}{Bebendorf03}%
\bibitem{Florenzano}%
  \BibitemOpen
  \bibfield{author}{%
  \bibinfo {author} {\bibfnamefont{M.}~\bibnamefont{Florenzano}}\ and\ \bibinfo
  {author} {\bibfnamefont{C.}~\bibnamefont{Le~Van}},\ }%
  \emph{\bibinfo {title} {Finite Dimensional Convexity and Optimization}}\
  (\bibinfo {publisher} {Springer Verlag},\ \bibinfo {address} {Berlin,
  Heidelberg, New York},\ \bibinfo {year} {2001})%
  \bibAnnoteFile{NoStop}{Florenzano}%
\bibitem{Rockafellar}%
  \BibitemOpen
  \bibfield{author}{%
  \bibinfo {author} {\bibfnamefont{H.~T.}\ \bibnamefont{Rockafellar}},\ }%
  \emph{\bibinfo {title} {Convex Analysis}}\ (\bibinfo {publisher} {Princeton
  University Press},\ \bibinfo {year} {1970})\ \bibinfo {note} {reprinted
  1996}%
  \bibAnnoteFile{NoStop}{Rockafellar}%
\end{thebibliography}
%

\end{document}